\documentclass{article}

\pdfoutput=1

\usepackage[margin=1in]{geometry}
\usepackage[utf8]{inputenc}
\usepackage{amsmath}
\usepackage{amssymb}
\usepackage{amsthm}
\usepackage{xargs}
\usepackage{bm}
\usepackage{graphicx}
\usepackage{booktabs}
\usepackage{tabularx}
\usepackage[colorlinks=true, linkcolor=blue]{hyperref}
\usepackage{scalerel,stackengine}
\usepackage{xcolor}
\usepackage{multirow}
\usepackage{makecell}
\usepackage{chngcntr}
\usepackage{apptools}
\usepackage{sidecap}
\usepackage{soul}
\usepackage{authblk}
\usepackage[numbers]{natbib}

\AtAppendix{\counterwithin{theorem}{section}}

\DeclareRobustCommand{\Fig}[1]{Figure~\ref{fig:#1}}
\DeclareRobustCommand{\Figs}[2]{Figure~\ref{fig:#1} and~\ref{fig:#2}}

\newtheorem{theorem}{Theorem}
\newtheorem{lemma}[theorem]{Lemma}

\theoremstyle{definition}

\theoremstyle{remark}

\newcommand{\numeral}[1]{%
  \textup{\uppercase\expandafter{\romannumeral#1}}%
}

\stackMath
\newcommand\reallywidehat[1]{%
\savestack{\tmpbox}{\stretchto{%
  \scaleto{%
    \scalerel*[\widthof{\ensuremath{#1}}]{\kern.1pt\mathchar"0362\kern.1pt}%
    {\rule{0ex}{\textheight}}
  }{\textheight}%
}{2.4ex}}%
\stackon[-6.9pt]{#1}{\tmpbox}%
}
\title{
Bayesian prognostic covariate adjustment
}

\begin{document}

\author[1]{David Walsh\thanks{dwalsh@unlearn.ai}}
\author[1]{Alejandro Schuler}
\author[1]{Diana Hall}
\author[1]{Jon Walsh}
\author[1]{Charles Fisher}
\affil[1]{Unlearn.AI, Inc., San Francisco, CA}

\author[]{for the Critical Path for Alzheimer's Disease\thanks{Data used in the preparation of this article were obtained from the Critical Path for Alzheimer's Disease (CPAD). As such, the investigators within CPAD contributed to the design and implementation of the CPAD database and/or provided data, but did not participate in the analysis of the data or the writing of this report.}}
\author[ ]{the Alzheimer's Disease Neuroimaging Initiative\thanks{Data used in preparation of this article were obtained from the Alzheimer’s Disease Neuroimaging Initiative (ADNI) database (\href{url}{adni.loni.usc.edu}). As such, the investigators within the ADNI contributed to the design and implementation of ADNI and/or provided data but did not participate in analysis or writing of this report. A complete listing of ADNI investigators can be found in \href{http://adni.loni.usc.edu/wp-content/uploads/how_to_apply/ADNI_Acknowledgement_List.pdf}{this document}.}}
\author[ ]{the Alzheimer's Disease Cooperative Study\thanks{Data used in preparation of this manuscript/publication/article were obtained from the University of California, San Diego Alzheimer’s Disease Cooperative Study. Consequently, the ADCS Core Directors contributed to the design and implementation of the ADCS and/or provided data but did not participate in analysis or writing of this report.}}

\date{December 2020}

\newcommand{\E}[1]{{\mathbb E} \left[#1\right]}
\newcommand{\V}[1]{{\mathbb V} \left[#1\right]}
\newcommand{\C}[1]{{\mathbb C} \left[#1\right]}
\newcommand{\iid}[0]{\overset{\text{IID}}{\sim}}
\newcommand{\Ehat}[1]{{\widehat{\mathbb E}} \left[#1\right]}
\newcommand{\Vhat}[1]{{\widehat{\mathbb V}} \left[#1\right]}
\newcommandx{\hattilde}[1]{\widehat{\widetilde{#1}}}
\newcommand{\nn}{\nonumber}

\maketitle

\begin{abstract}
Historical data about disease outcomes can be integrated into the analysis of clinical trials in many ways.  We build on existing literature that uses prognostic scores from a predictive model to increase the efficiency of treatment effect estimates via covariate adjustment.  Here we go further, utilizing a Bayesian framework that combines prognostic covariate adjustment with an empirical prior distribution learned from the predictive performances of the prognostic model on past trials. The Bayesian approach interpolates between prognostic covariate adjustment with strict type I error control when the prior is diffuse, and a single-arm trial when the prior is sharply peaked. This method is shown theoretically to offer a substantial increase in statistical power, while limiting the type I error rate under reasonable conditions. We demonstrate the utility of our method in simulations and with an analysis of a past Alzheimer's disease clinical trial.

\end{abstract}

\section{Introduction}
\label{sec:intro}

Randomized controlled trials (RCTs) are the gold standard for clinical research, because randomization blunts the impact of unobserved confounders \cite{Sox:2012hu, Overhage:2013fx, Hannan:2008gh}. The estimation of any treatment effects still carries some statistical uncertainty, but this may be reduced to a stated tolerance on the \textit{type I error} rate. Safe in that knowledge, practitioners may optimize the designs of their trials so as to maximize statistical \textit{power}.

If we consider the efficacy of the treatment as fixed, the primary determinants of statistical power are the sample sizes. Thus, the most straightforward method to increase power is to run a larger trial that randomizes more subjects into both the active treatment and placebo arms. However, trial costs and timelines typically scale with the number of subjects, making large trials economically and logistically challenging. Moreover, for diseases with high mortality, it can be unethical to recruit large numbers into the placebo arm \cite{FDA_placebo_cancer}.

Given these challenges, it is tempting to leave the RCT framework behind. So-called ``single-arm'' trials recruit patients only for the active treatment arm, and rely on external sources of placebo data such as electronic health records, patient registries, insurance claims, and clinical trials that have been run previously. Causal inference approaches such as matching or propensity score weighting may be used to construct a ``synthetic control group'', which is demographically similar to the trial subjects \cite{berry2017creating, abadie2010synthetic}, from these external sources. Alternatively, ``prognostic models'' are gaining popularity as another way of eliminating concurrent placebo groups \cite{Hansen:2008cw, aikens2019using, Wyss:2014ef}. These are machine learning models that have been trained on the external data to predict the trial subjects' hypothetical outcomes under the placebo, which we call their ``prognostic scores''. A treatment effect can be estimated by subtracting these prognostic scores from observed outcomes on the active treatment. Either way, these approaches allow for the estimation of treatment effects with only half as many subjects, none of whom must receive a placebo. However, one must make strict (often unverifiable) assumptions about which demographic covariates are important, or on the accuracy of the prognostic model, to ensure that results obtained using these methods are robust to unobserved confounders. Given the sensitivity to unknown-unknowns, single-arm trials are usually inappropriate for pivotal trials used to inform regulatory decisions, except in cases of dire need \cite{FDA_control_group}.

\subsection{Prognostic covariate adjustment}
\label{sec:intro_prog}

In contrast to single-arm trials with external controls, a variety of approaches that fall under the umbrella of ``historical borrowing'' aim to leverage information from external data sources to improve statistical power in clinical trials without eliminating randomization altogether. Such methods offer a compromise -- they aim to increase statistical power without adding too much bias. 

Recently, Schuler et al. \cite{freq-methods} proposed a method that exploits prognostic models to increase power while maintaining {\it strict} type I error rate control in the large-sample setting. They propose to run an RCT, in which a prognostic model trained on historical data is used to compute a prognostic score for each subject in the trial. Then, the prognostic scores are adjusted for as a covariate in a linear regression to estimate the treatment effect. This estimator maintains strict (asymptotic) control of the type I error rate even if the prognostic model is inaccurate, yet has minimum variance of all unbiased estimators if the prognostic model is perfect and the treatment effect is constant. Strict control of type I error is a major strength of prognostic covariate adjustment, but it limits the amount of power that could be gained over an ``unadjusted'' analysis (one that does not incorporate historical information at all).

Our novel method - Bayesian prognostic covariate adjustment - is a Bayesian analysis that draws on the strengths of the prognostic model approach. Where Schuler et al. \cite{freq-methods} use prognostic scores only as covariates on each trial subject, we recognize that these scores store historical information about population-level expected outcomes, which we can formalize as a prior distribution on parameters in the linear regression. Bayesian prognostic covariate adjustment interpolates between their frequentist approach, which makes no assumptions on the prognostic model, and a single-arm trial, which treats the prognostic scores as true placebo outcomes. Indeed, the strength of our chosen prior distribution is controlled through a single tuning parameter $\lambda$, and we recover these two cases at extreme values of $\lambda$.

Although Bayesian approaches are becoming more popular for clinical trials, the operating characteristics such as the type I and type II error rates remain important properties of clinical trials used to support regulatory decisions \cite{FDA_bayes}. Therefore, despite proposing a Bayesian method, we will analyze its frequentist properties both theoretically and empirically. We show how it achieves better power than prognostic covariate adjustment alone, while ensuring type I error control under weak assumptions.

\subsection{Bayesian clinical trials}
\label{sec:intro_bayes}

Bayesian approaches to historical borrowing in clinical trials are far from new. For instance, previous work has proposed approaches to fitting prior distributions for placebo outcomes directly from electronic health records~\cite{khozin2017real, berger2016opportunities}. In addition to the power gains, the appeal of these Bayesian methods is their straight-forward interpretability, as well the wealth of information contained in a posterior distribution that is absent from a frequentist p-value \cite{jack2012bayesian}.

One challenge is that, like single-arm trials, most of the Bayesian literature relies on unverifiable assumptions about the prior distribution to ensure robustness to confounding. For example, Psioda and Ibrahim \cite{psioda2019bayesian} note how treatment effects can be mis-estimated when the patients measured in the external data sources differ from the trial population in terms of their demographics or disease status. Even when the external data is relevant, choices about which clinical trials to run can introduce selection biases to the pool of relevant data~\cite{howard2005bayesian}. As with single-arm trials, some Bayesian methods attack confounding with causal inference techniques \cite{wang2019bayesian, lin2019propensity}. Others quantify the impact of the prior, assigning to it an ``effective sample size'' \cite{morita2010evaluating, wiesenfarth2020quantification}. When this effective sample size threatens to overwhelm the trial data, ``power priors'' may be used to downweight the external information until more desirable operating characteristics are obtained~\cite{ibrahim2000power,duan2006evaluating, hobbs2011hierarchical, schmidli2014robust}. However, these methods generally cannot increase power without inflating type I error rates~\cite{kopp2020}.

There is a large literature debating the relative merits of the Bayesian and frequentist approaches to inference. Regardless of the arguments for or against Bayesian methods in general, there is a practical matter of regulatory guidance for the analyses of clinical trials investigating new drugs and medical devices. Guidance from the United States Food and Drug Administration (FDA) specifically notes the need for a binary conclusion about a treatment's efficacy using a decision rule with ``reasonable control of the type I error rate''~\cite{FDA_bayes}. Many researchers have presented Bayesian definitions of type I error rates, which usually involve taking an average over a ``null'' prior distribution in which the treatment is deemed functionally ineffective~\cite{spiegelhalter1986predictive, brown1987projection, chen2011bayesian}. However, these do not directly address regulatory concerns about the frequentist properties of the decision rule. 

In contrast to other Bayesian methods, we provide reasonable control of the type I error rate under weak assumptions. Incorporating historical information through a prognostic score allows us to reduce the choice of prior distribution to a single parameter, $\lambda$, which reflects the anticipated magnitude of any bias in the prognostic model. The simplicity of this approach enables us to derive explicit formulae for the type I and type II error rates of a Bayesian decision rule. Moreover, the single parameter of the prior distribution can be estimated by measuring the prognostic model's performance on historical clinical trials. If $\lambda$ is chosen appropriately, we can control type I error under the strict null hypothesis of a zero treatment effect.

\subsection{Our contributions}
\label{sec:intro_summary}

We begin with the linear prognostic covariate adjustment formulation presented in Schuler et al.~\cite{freq-methods}, but reparameterize the linear model so that the intercept term, $\beta_0$, and the variance, $\sigma^2$, may be interpreted respectively as the bias and the residual variance of the prognostic model at predicting placebo outcomes for each trial patient. With this formulation, we focus on a single dimension, $\beta_0/\sigma$, that characterizes the average bias of the prognostic model scaled by the unexplained variation. We define a Normal Inverse-Gamma prior parameterized by a single parameter, $\lambda$, which represents an anticipated upper bound on the value of $\beta_0/\sigma$. That is, small values of $\lambda$ reflect a prognostic model with little bias, whereas large values of $\lambda$ reflect a prognostic model that may have a large bias. All other parameters in the regression are given non-informative priors. We prove that if $\lambda$ is chosen successfully, approximate type I error control is achieved.

Furthermore, we show that $n\lambda^2$ defines a natural scale for interpolating between the prognostic covariate adjustment method of Schuler et al.~\cite{freq-methods} and a single-arm trial. Specifically, we analyze Bayesian prognostic covariate adjustment in an asymptotic regime in which the trial sample size $n$ approaches infinity, and an abundance of historical data lets us take $\lambda \to 0$. When $n\lambda^2$ is large, the trial data overwhelms the prior and our method reduces to the prognostic covariate adjustment analysis with strict type I error rate control. On the other hand, as $n\lambda^2$ falls towards zero, we approach the single-arm scenario in which the prognostic score for each patient in the active treatment arm is assumed to represent their potential outcome in the counterfactual where they received a placebo. This single-arm limit maximizes power, but will inflate the type I error rate if the prognostic model is not perfect.

In this paper we define the analysis in Section~\ref{sec:methods} and explore the asymptotics in Section~\ref{sec:theory}.  We demonstrate that, in the asymptotic case of a large number of subjects and a narrow prior (small $\lambda$), the variance of the Bayesian estimator has a simple relationship to the frequentist prognostic covariate adjustment estimator. An empirical analysis of the theory is performed through simulations and the reanalysis of a past clinical trial in Alzheimer's disease in Section~\ref{sec:demonstration}. We close with a discussion of results in Section~\ref{sec:discussion}.

\section{Background}
\label{sec:prognostic}

This section reviews frequentist approaches that leverage prognostic scores. In Section \ref{sec:target}, we define notation for the prognostic model and how it may be applied to the ``target trial'': the study where we hope to improve inference through historical information. In Section \ref{sec:freq_summary}, we summarize the prognostic covariate adjustment work by Schuler et al.~\cite{freq-methods}, while in Section \ref{sec:single-arm} we define the opposite extreme, where the prognostic scores are assumed to represent placebo outcomes exactly. We define these two paradigms here, because in the remainder of this paper we will show how a Bayesian analysis can leverage additional historical information to interpolate between them, achieving a more practically meaningful tradeoff of type I error and power. As a benchmark for all historical borrowing methods, we define the ``unadjusted'' analysis in Section \ref{sec:unadjusted}, which disregards the prognostic scores entirely.

\subsection{Setup}
\label{sec:target}

For notational simplicity, we begin with some assumptions about the target:
\begin{itemize}
    \item The $n$ subjects are randomized such that a deterministic number, $pn$, are assigned to the active treatment group, while the remaining $(1-p)n$ subjects will receive the placebo. For each patient, a binary random variable $W$ indicates their treatment assignment ($W = 0$ for placebo and $W=1$ for active treatment). Thus the collection $(W_1, \dots, W_n)$ is sampled uniformly from the set of permutations of $pn$ ones and $(1-p)n$ zeros.
    \item A single, real-valued outcome, $Y$, is measured for each subject.
    \item Both prognostic covariate adjustment and the single-arm analysis leverage a fitted prognostic model, $\mathcal{M}$, which has been trained on historical placebo data, to optimize inference on the target trial. We are agnostic to the exact procedure used to fit this model. The prognostic model outputs a prognostic score, $M$, for each subject in the target trial. This score may be interpreted as a prediction of the patient's outcome, $Y$, given their baseline characteristics, in the event that $W=0$. We will assume that the prognostic model is deterministic, given its inputs. However, each $M$ is a random variable, reflecting variability in the baseline covariates of patients entering the study. We will assume that these variables are IID draws from some non-degenerate, square integrable distribution that is specific to the trial. These variables are independent of the treatment assignments, $W$.
    \item The disease progression of each subject is stochastic, with a distribution that depends on their baseline characteristics and their treatment assignments. Conditional on all $M_1, \dots, M_n, W_1, \dots, W_n$, we model the outcomes, $Y_1, \dots, Y_n$, as independent random variables. Their common distribution is given by:
    \begin{equation}
        \label{eq:lin_mod}
        Y_i | M_1, \dots, M_n, W_1, \dots, W_n \sim \beta_0 + \beta_1 W_i + \beta_2 (M_i - \bar{M}) + \bar{M} + N(0,\sigma^2)
    \end{equation}
    where $\bar{M} = \frac{1}{n}\sum_{i=1}^n M_i$ and $\beta_0, \beta_1, \beta_2$ and $\sigma^2$ are unknown parameters.
\end{itemize}

Within this structural model, we interpret
\begin{align}
    \label{eq:beta_0}
    \beta_0 &= E(Y_i|W_i = 0, M_1, \dots, M_n) - \bar{M}\\
    \label{eq:sigma^2}
    \sigma^2 &= Var(Y_i|W_i = 0, M_1, \dots, M_n),
\end{align}
respectively as the average bias and the residual variance of the prognostic model, $\mathcal{M}$, when applied to the target trial. If the prognostic model estimates the conditional means of $Y$ accurately, we would expect the bias to be small. On the other hand, the stochasticity of disease progression places a floor on the residual variance that can be achieved. If a patient's baseline characteristics are perfectly predictive of their outcomes under placebo, one may construct a prognostic model with $\sigma^2 = 0$ and $\beta_2 = 1$. However, this is not possible in general.

Our structural model assumes a constant treatment effect, $\beta_1 = E(Y_i|M_i, W_i = 1) - E(Y_i|M_i, W_i = 0)$, in the sense that this quantity does not depend on the subject's baseline characteristics through their value of $M$. This effect is the target of inference.

\subsection{Prognostic covariate adjustment}
\label{sec:freq_summary}

Our choice of structural model is inspired by Schuler et al.~\cite{freq-methods}, who establish that linear regression of subject outcomes on prognostic scores achieves the optimal noise reduction, provided that the prognostic model can recover the conditional mean of $Y$ given each subject's baseline characteristics and that the treatment effect is constant. The analysis method proposed in Schuler et al.~\cite{freq-methods}, which we will refer to simply as ``prognostic covariate adjustment'', is to fit all unknown parameters in \eqref{eq:lin_mod} by maximum likelihood. For a given significance level $\alpha$, the treatment is deemed effective if a t-test rejects the null hypothesis that $\beta_1 = 0$.

The following result follows immediately from Corollary A.4.1 in Schuler et al.~\cite{freq-methods}, after noting that their definitions of $\sigma_0^2(1-\rho_0^2)$ and $\sigma_1^2(1-\rho_1^2)$ both coincide with $\sigma^2$ under the assumptions of this paper.
\begin{lemma}
\label{lemma:prog_power}
After fixing true values for $\beta_0, \beta_1, \beta_2$ and $\sigma^2$, and letting $n \to \infty$, the power of prognostic covariate adjustment approaches:
\begin{equation}
\label{eq:prog_power}
\Phi \left ( \Phi^{-1}(\alpha/2) + \frac{\beta_1 \sqrt{np(1-p)}}{\sigma} \right ) + \Phi \left ( \Phi^{-1}(\alpha/2) - \frac{\beta_1 \sqrt{np(1-p)}}{\sigma} \right )
\end{equation}
\end{lemma}

\subsection{Single-arm analyses}
\label{sec:single-arm}

Compared with prognostic covariate adjustment, a more aggressive approach is to assume that the prognostic model recovers each subject's expected outcome under the placebo precisely. In view of \eqref{eq:lin_mod} evaluated at $W_i = 0$, this is the assumption that $\beta_0=0$ and $\beta_2=1$. This in turn implies that:
\begin{equation}
        \label{eq:lin_mod_single}
        Y_i | M_i, \{W_i = 1\} \sim \beta_1 + M_i + N(0,\sigma^2)
    \end{equation}
    
We define the ``single-arm'' analysis as a linear regression of the differences $(Y-M)$ for the active treatment arm only on an intercept term, $\beta_1$, which is fitted by maximum likelihood. The treatment is deemed effective if a t-test rejects the null hypothesis that $\beta_1 = 0$ at a given significance level, $\alpha$.

The following result is proven in Appendix \ref{sec:proofs}.
\begin{lemma}
After fixing true values for $\beta_0, \beta_1, \beta_2$ and $\sigma^2$, and letting $n \to \infty$, the power of the single-arm analysis approaches:
\label{lemma:single_power}
\begin{equation}
\label{eq:single_power}
    \Phi \left ( \Phi^{-1}(\alpha/2) + \frac{(\beta_1 + \beta_0) \sqrt{pn}}{\sigma} \right ) + \Phi \left ( \Phi^{-1}(\alpha/2) - \frac{(\beta_1 + \beta_0) \sqrt{pn}}{\sigma} \right )
\end{equation}
\end{lemma}

\subsection{Unadjusted analyses}
\label{sec:unadjusted}

In contrast to prognostic covariate adjustment, which fits eq. \eqref{eq:lin_mod}, an ``unadjusted analysis'' ignores all subjects' prognostic scores. This approach considers a linear regression of $Y$ on $W$ and an intercept term only. This naive regression is fitted by maximum likelihood, and the treatment is deemed effective if a t-test at level $\alpha$ rejects the null hypothesis that the coefficient on $W$ is zero.

\section{Our proposal}
\label{sec:methods}

Our approach is to fit the structural model given by eq. \ref{eq:lin_mod} with a Bayesian linear model. Our joint prior for $\beta_0, \beta_1, \beta_2, \sigma^2$ is taken from the conjugate Normal Inverse-Gamma family for computational convenience \cite{banerjee2008bayesian}. Our main objective with this prior is to encode a belief that $|\beta_0/\sigma| \leq \lambda$; i.e. that the ratio of the average bias of the prognostic model to its residual variance is less than a parameter $\lambda$. This parameter therefore describes our confidence in the predictive accuracy of our prognostic model. We seek to place non-informative priors on all other parameters. Naturally we want a non-informative prior for $\beta_1/\sigma$, because historical information about the active treatment of the target trial is not generally available. It is possible that historical trials could be used to learn a reliable, informative prior on $\beta_2/\sigma$. However, we do not pursue such an approach, as it offers little added benefit in terms of power (see Theorem \ref{thm:beta2} in the appendix).

Based on these considerations, we define our prior as follows. For each $\epsilon > 0$, let $\mathcal{P}_\epsilon$ denote the following joint distribution, where \eqref{eq:sigma_eps} and \eqref{eq:beta_eps} hold independently:
\begin{align}
    \label{eq:sigma_eps}
    \sigma^2 &\sim Inverse-Gamma(\epsilon, \epsilon)\\
    \label{eq:beta_eps}
    \frac{1}{\sigma} \begin{pmatrix}
    \beta_0\\
    \beta_1\\
    \beta_2
    \end{pmatrix} &\sim N \left ( 0, \begin{pmatrix} \lambda^2 & 0 & 0\\
    0 & 1/\epsilon & 0\\
    0 & 0 & 1/\epsilon
    \end{pmatrix} \right )
\end{align}
Then starting from the prior $\mathcal{P}_\epsilon$, let $\mathcal{P}'_\epsilon$ denote the joint posterior for all parameters, after all data has been observed for the target trial. Since it depends on the data, $\mathcal{P}'_\epsilon$ is a random variable, which takes values in the space of Normal Inverse-Gamma distributions over $\mathbb{R}^+ \times \mathbb{R}^3$. In view of the following lemma, whose proof is given in Appendix \ref{sec:proofs}, we propose to analyze the data using the limiting posterior distribution, $\mathcal{P}'$.

\begin{figure}[hbt!]
\begin{center}
\includegraphics[width=4.5in]{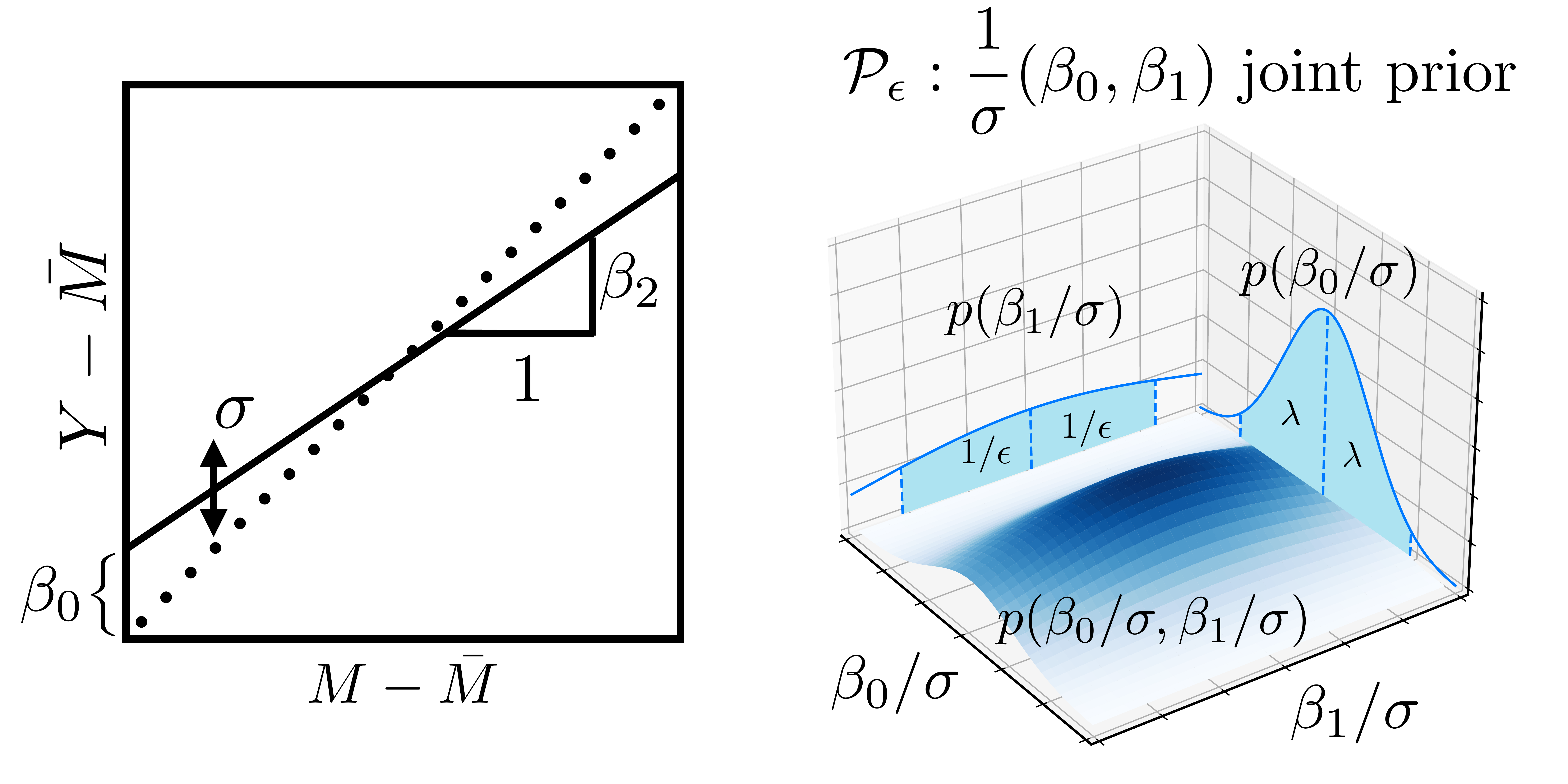}
\end{center}
\caption{\textbf{Left: \boldsymbol{$\beta_0$} captures bias in M}, the prognostic score, for a subject in the placebo arm (i.e. when $W=0$). $\beta_2$ captures how the bias varies with M, shown as the slope of the line. $\sigma^2$ represents the variance of outcome $Y$, conditional on $M$ and $W$. The dotted line shows the case where $\beta_2$=1. \textbf{Right: \boldsymbol{$\lambda$} controls the confidence in M}, shown in the joint prior of the treatment effect $\beta_1$ and bias $\beta_0$. $\beta_1 / \sigma$ has a flat prior indicating no knowledge of the treatment effect $\beta_1$ at the outset of the target trial, similar to a frequentist framework.
\label{fig:model_params}}
\end{figure}

\begin{lemma}
\label{lemma:post}
Let $\mathbf{Y}$ be the vector of centered outcomes, $Y_i - \bar{M}$, across all $n$ subjects, let $\mathbf{X}$ be the matrix of row vectors, $(1, W_i - p, M_i - \bar{M})^T$, and define the random variables:
\begin{align}
    \label{eq:V}
    V^{-1} &= \begin{pmatrix} 1/\lambda^2 & -p/\lambda^2 & 0\\
    -p/\lambda^2 & p^2/\lambda^2 & 0\\
    0 & 0 & 0
    \end{pmatrix} + \mathbf{X}^T \mathbf{X},\\
    \label{eq:mu}
    \mu &= V \mathbf{X}^T \mathbf{Y}, \\
    \label{eq:S^2}
    S^2 &= \mathbf{Y}^T \mathbf{Y} - \mu^T V^{-1} \mu.
\end{align}
Note that these are well-defined, because the given expression for $V^{-1}$ is invertible almost surely.

Then, almost surely as $\epsilon \to 0$, $\mathcal{P}'_\epsilon$ converges to a limiting Normal Inverse-Gamma distribution, $\mathcal{P}'$, which is defined as follows:
\begin{align}
    \label{eq:sigma_post}
    \sigma^2 &\sim Inverse-Gamma(n/2, S^2/2)\\
    \label{eq:beta_post}
    \begin{pmatrix}
    \beta_0 + p \beta_1\\
    \beta_1\\
    \beta_2
    \end{pmatrix} | \, \sigma^2 &\sim N(\mu, \sigma^2 V)
\end{align}
Under $\mathcal{P}'$,
\begin{equation}
    \label{eq:beta1_post}
    \frac{\beta_1 - \mu_1}{\sqrt{V_{11} S^2 / n}} \sim t_n
\end{equation}
\end{lemma}

We interpret the posterior mean, $\mu_1$, as a frequentist estimator of $\beta_1$. When $n \lambda^2$ is large, the prior is flat and our estimator is identical to prognostic covariate adjustment. However, at smaller $\lambda$, we obtain an estimator with lower variance by shrinking our estimate of the expected placebo outcome ($\beta_0 + \bar{M}$) towards $\bar{M}$.

While there are many ways that the posterior for $\beta_1$ could be used to assess the efficacy of the treatment, we focus on a binary decision rule. For a given signficance level, $\alpha$, we deem the treatment effective if the posterior probability on $\{\beta_1 > 0\}$ either exceeds $(1-\alpha/2)$ or falls below $\alpha/2$. This decision rule mirrors 2-tailed frequentist hypothesis testing, so we will refer to this decision as a ``rejection of the null hypothesis''. For a more detailed discussion on ``probability of sign'' rules and their connection to frequentist testing, see Makowski et al. \cite{makowski2019indices}.

From Lemma \ref{lemma:post}, we see that a rejection is obtained when
\begin{equation}
    \label{eq:reject}
    \frac{\mu_1}{\sqrt{V_{11} S^2 / n}} > (F^t_n)^{-1}(1 - \alpha/2)
\end{equation}
where $F^t_n$ denotes the distribution function of a $t_n$ random variable. We will refer to \eqref{eq:reject} as the ``rejection event''. The frequency with which this event occurs is the ``rejection rate'' of the null hypothesis, which corresponds to either power or type I error rate depending on the veracity of the null.

\subsection{Choosing the prior}
\label{sec:hist}

For the key results of this paper, we are agnostic about how the historical trials are used to construct the upper bound $\lambda$ for $|\beta_0|/\sigma$; we will simply establish that our method offers approximate type I error control when $\lambda$ is chosen appropriately.

We note here that choosing $\lambda$ is reasonably trivial if all of the following conditions hold:
\begin{enumerate}
    \item The number of historical trials, $m$, is large.
    \item The sample sizes of each historical trial are large.
    \item The control arms of the historical trials are relevant to that of the target (see e.g. \cite{pocock1976combination} for a formal discussion of ``relevance'').
\end{enumerate}
In that case, for $j = 1, \dots, m$, the bias $\beta_{0,j}$ and residual variance $\sigma_j^2$ of the prognostic model $\mathcal{M}$ applied to historical trial $j$ can be defined by analogy to equations \eqref{eq:beta_0} and \eqref{eq:sigma^2}.\footnote{Here we assume that the residual variance is constant across the subjects within individual historical trials, so that $\sigma_j^2$ is well-defined.} Each ratio, $\beta_{0,j}/\sigma_j$, may be estimated with little error, and we might take $\lambda$ as e.g. the upper 95th percentile of the values: $|\beta_{0,1}|/\sigma_1, \dots, |\beta_{0,m}|/\sigma_m$. Section \ref{sec:lambda} addresses practical considerations for choosing $\lambda$ when historical data is less plentiful.

\section{Asymptotic behavior}
\label{sec:theory}

In this section, we analyze the asymptotic rejection rate of our Bayesian analysis in the limit where $n \to \infty$\footnote{
We focus on the given asymptotic with $n\lambda^2 = O(1)$, because it ensures that the ratio of information between the target trial and the prior remains constant. Under any of the distributions, $\mathcal{P}_\epsilon$, $1/\lambda^2$ is chosen to be the prior precision on $\beta_0/\sigma$. Heuristically we may view $1/\lambda^2$ as an effective sample size (see e.g. \cite{morita2008determining} for details), and so $n\lambda^2$ represents a ratio of sample sizes between the target trial and our chosen prior.
}
and $\lambda \to 0$, such that $n\lambda^2 = O(1)$. We find that the power of our method varies between the power of prognostic covariate adjustment and the greater power of the single-arm estimator, depending on the limiting value of $n \lambda^2$. But in contrast to the single-arm estimator, if $\lambda$ is chosen successfully as an upper bound on $|\beta_0/\sigma|$, then the type I error is controlled approximately.

We begin with the main theoretical contribution of this paper: an asymptotic expression for the rejection rate in terms of the true values of all unknown parameters. The proof of this theorem is given in Appendix \ref{sec:proofs}. Evaluating this expression at $\beta_1=0$ gives an asymptotic type I error rate for our Bayesian method, while $\beta_1 \neq 0$ describes its power.
\begin{theorem}
\label{thm:rejection_rate}
Fix the true values of $\beta_0, \beta_1, \beta_2$ and $\sigma^2$, and consider the limit where $\lambda^2 \to 0$ and $n \to \infty$ such that $n\lambda^2 = O(1)$.
\begin{enumerate}
    \item
\begin{equation}
    \label{eq:V_11}
    n V_{11} \to_P \frac{n \lambda^2 + 1}{n \lambda^2p(1-p) + p}.
\end{equation}
\item To leading order, the rejection rate is given by:
\begin{equation}
\label{eq:rejection_rate}
\begin{split}
    \Phi \left ( \Phi^{-1}(\alpha/2) \sqrt{\frac{V_{11}}{\hat{V}} \left (1 + \frac{(1-p)\beta_0^2}{\sigma^2 (n\lambda^2(1-p) + 1)} \right )} + \frac{\tau}{\sigma \sqrt{\hat{V}}} \right )\\
    + \Phi \left ( \Phi^{-1}(\alpha/2) \sqrt{\frac{V_{11}}{\hat{V}} \left (1 + \frac{(1-p)\beta_0^2}{\sigma^2 (n\lambda^2(1-p) + 1)} \right )} - \frac{\tau}{\sigma \sqrt{\hat{V}}} \right )
\end{split}
\end{equation}
where
\begin{align}
    \label{eq:tau}
    \tau &= \beta_1 + \left ( \frac{1}{n\lambda^2(1-p) + 1} \right ) \beta_0, \\
    \label{eq:V_hat}
    \hat{V} &= \frac{p + (1-p)(n\lambda^2 + 1)^2}{np(n \lambda^2(1-p) + 1)^2}.
\end{align}
and $\Phi$ is the CDF of the standard normal distribution.
\end{enumerate}
\end{theorem}

In the proof of Theorem \ref{thm:rejection_rate}, we show that the asymptotic sampling distribution of $\mu_1$ approaches $N(\tau, \sigma^2 \hat{V})$. Noting that the prognostic covariate adjustment estimator has asymptotic variance, $\sigma^2/np(1-p)$, we see that the sampling variance has been scaled by a factor of
\begin{equation}
\label{eq:factor}
\frac{p(1-p) + (1-p)^2(n\lambda^2 + 1)^2}{(n \lambda^2(1-p) + 1)^2},
\end{equation}
which is always less than one. As $n \lambda^2 \to 0$, $\sigma^2 \hat{V}$ approaches the sampling variance of the single-arm estimator: $\sigma^2/np$.

This is shown in Figure~\ref{fig:type1_2_errs_example} in terms of the power and type I error rate for an example scenario.  As $\lambda$ varies, the operating characteristics of the Bayesian decision rule interpolate between the prognostic covariate adjustment and single-arm estimators, allowing practitioners to achieve a desired tradeoff between power and type I error depending on the application.  The figure shows this tradeoff in operating characteristics for one particular value of $n$, as well as the power dependence on $n$.  These power and type I error rate values come from theoretical calculations at points in parameter space that have $\beta_0 > 0$ and differ only in the value of $\beta_1$.  Additionally, this example suggests that the Bayesian estimator can achieve a highly desirable middle ground between the two estimators, with points that have significantly higher power than the prognostic covariate adjustment estimator at a minimal increase of type I error.

\begin{figure}[hbt!]
\begin{center}
\includegraphics[width=6.5in]{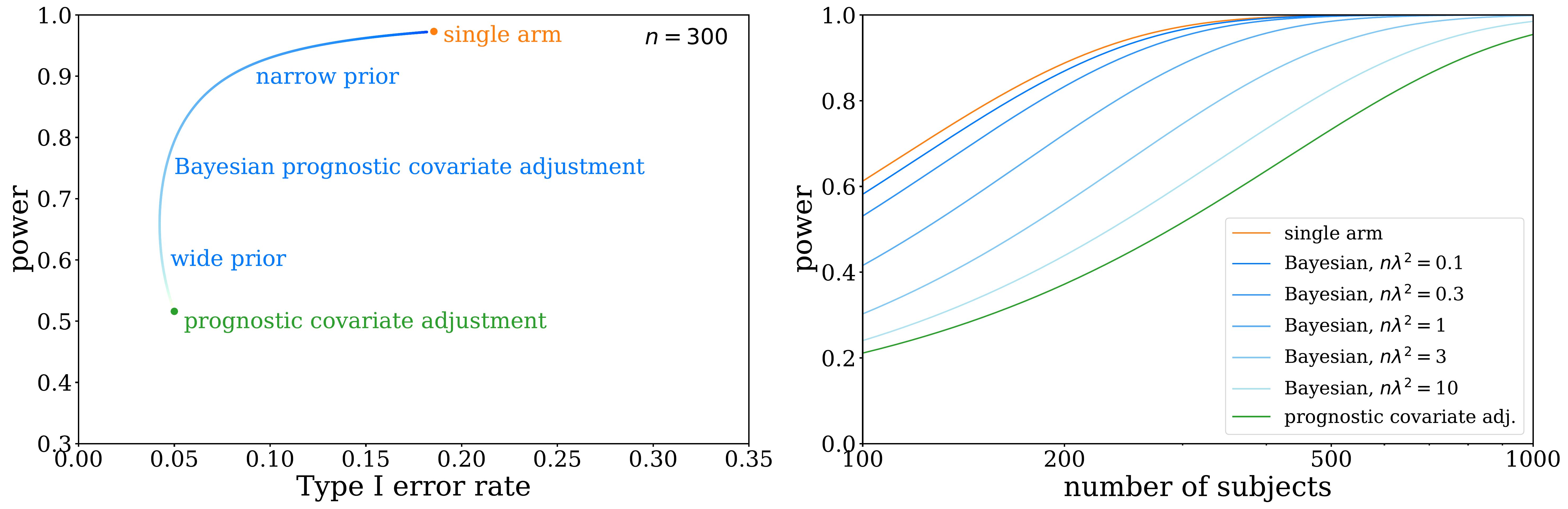}
\end{center}
\caption{{\bf The operating characteristics of the Bayesian decision rule interpolate between two frequentist estimators.} On the left, the theoretical type I error rate and power for Bayesian prognostic covariate adjustment are shown as the prior variance changes, compared to the same metrics for the prognostic covariate adjustment and single-arm estimators. The type I error rate and power are determined for two points in parameter space that differ only by the value of $\beta_1$, where both points have $\beta_0 > 0$.  The curve shows the Bayesian estimator's operating characteristics as a function of $\lambda$ ranging from $n \lambda^2 \ll 1$ to $n \lambda^2 \gg 1$.  On the right, the power for the $\beta_1 > 0$ point in parameter space is shown as $n$ varies for the Bayesian estimator for different values of the prior variance, as well as for the prognostic covariate adjustment and single-arm estimators.
\label{fig:type1_2_errs_example}}
\end{figure}

\subsection{Power}

To understand the power of our method, it is useful to examine \eqref{eq:rejection_rate} at extreme limiting values of $n \lambda^2$. When $n \lambda^2 \to \infty$, this expression approaches the asymptotic power of prognostic covariate adjustment previously given in \ref{eq:prog_power}.

For other values of $n\lambda^2$, the power of our Bayesian estimator exceeds that of prognostic covariate adjustment. Maximum power is attained as $n \lambda^2 \to 0$, where \eqref{eq:rejection_rate} approaches:
\begin{equation}
\label{eq:rejection_rate_zero}
    \Phi \left ( \Phi^{-1}(\alpha/2) \sqrt{1 + \frac{(1-p)\beta_0^2}{\sigma^2}} + \frac{(\beta_1 + \beta_0) \sqrt{np}}{\sigma} \right ) + \Phi \left ( \Phi^{-1}(\alpha/2) \sqrt{1 + \frac{(1-p)\beta_0^2}{\sigma^2}} - \frac{(\beta_1 + \beta_0) \sqrt{np}}{\sigma} \right )
\end{equation}
This coincides roughly with the power of the single-arm estimator as previously given in \eqref{eq:single_power}. 

The slight discrepancy is due to the different ways that the Bayesian and single-arm approaches estimate $\sigma^2$. Because the single-arm analysis uses only data from the active treatment arm, it obtains a consistent estimate of $\sigma^2$. On the other hand, at small $\lambda$, the Bayesian approach attributes to this residual variance any bias in the prognostic model that is observed on the target placebo arm, leading to a posterior mean for $\sigma^2$ that is inflated by the true value of $[1 + (1-p)\beta_0^2/\sigma^2]$ (this fact is apparent in the proof of Theorem \ref{thm:rejection_rate}).

\subsection{Type I Error}

The type I error rate of our method is given by evaluating \eqref{eq:rejection_rate} at $\beta_1=0$. Plugging that in shows that the asymptotic type I error rate only depends on the unknown parameters through the absolute value of $\beta_0/\sigma$. The condition $|\beta_0/\sigma| \leq \lambda$ ensures the asymptotic type I error rate is approximately at most $\alpha$. Although our method achieves this theoretical rate only at limiting values of $n$ and $\lambda$, the numerical analysis presented in Appendix~\ref{sec:visual} shows that the discrepancy at finite $n$ is small, so this condition is broadly sufficient.

\section{Empirical performance}
\label{sec:demonstration}

In this section we validate the Bayesian prognostic covariate adjustment theory through simulations, comparing theoretical predictions to simulations in terms of power and type I error.  We then demonstrate the empirical performance of the estimator by analyzing a past Alzheimer's disease clinical trial, comparing results to those obtained with other estimators. Further visualizations of the asymptotic theory are available in Appendix~\ref{sec:visual}.

\subsection{Simulations}
\label{sec:simulations}

To validate the operating characteristics defined by the framework presented here, we conduct a set of simulation studies.  Simulations are performed using a generative process that falls within the assumptions of the theory:
\begin{align}
\label{eq:sims_dgp}
    M &\sim {\cal N} (0, 1) \,, \\
    Y_C \,|\, M &\sim \beta_0 + \beta_2 M + \sigma {\cal N} (0, 1) \,, \nn \\
    Y_T \,|\, M &\sim \beta_0 + \beta_1 + \beta_2 M + \sigma {\cal N} (0, 1) \,. \nn
\end{align}
A nonlinear relationship between $M$ and $Y$ is also explored in the appendix.  Simulations are performed by repeatedly sampling from the generative processes with particular parameters, evaluating the posterior distribution each time, and computing the rate at which the null hypothesis is rejected (see eq. \ref{eq:reject}). This rejection rate defines the type I error rate when $\beta_1 = 0$, or the power when $\beta_1 > 0$.

There are several parameters that must be defined for any simulation: $\beta_0, \beta_1, \beta_2$, $\sigma$, the number of subjects $n$, and the standard deviation $\lambda$ of the prior on $\beta_0/\sigma$.  All simulations in this manuscript occur around variations of the following point:
\begin{align}
    \beta_0 = 0,\quad \beta_1 = 0,\quad \beta_2 = 1,\quad \sigma = \sqrt{3},\quad n = 1000,\quad n\lambda^2 = 1\,.
\end{align}
For type I error and power, different variations around these parameters are of interest.  Note that $\sigma$ is directly related to the correlation between the prognostic scores and the true outcomes; under \eqref{eq:sims_dgp}, the choice of $\sigma = \sqrt{3}$ implies a correlation of 1/2.\footnote{For the nonlinear case in the appendix, this choice of $\sigma$ produces a correlation of $1/\sqrt{2}$.}

\subsubsection{Type I Error}
\label{sec:simtypeIerr}

Validating the regime of type I error control -- i.e. where the error rate is less than the significance level, $\alpha$ -- is crucial.  As indicated above, this regime should correspond to $\beta_0 / \sigma < \lambda$.

\begin{figure}[htp!]
\begin{center}
\includegraphics[width=5.5in]{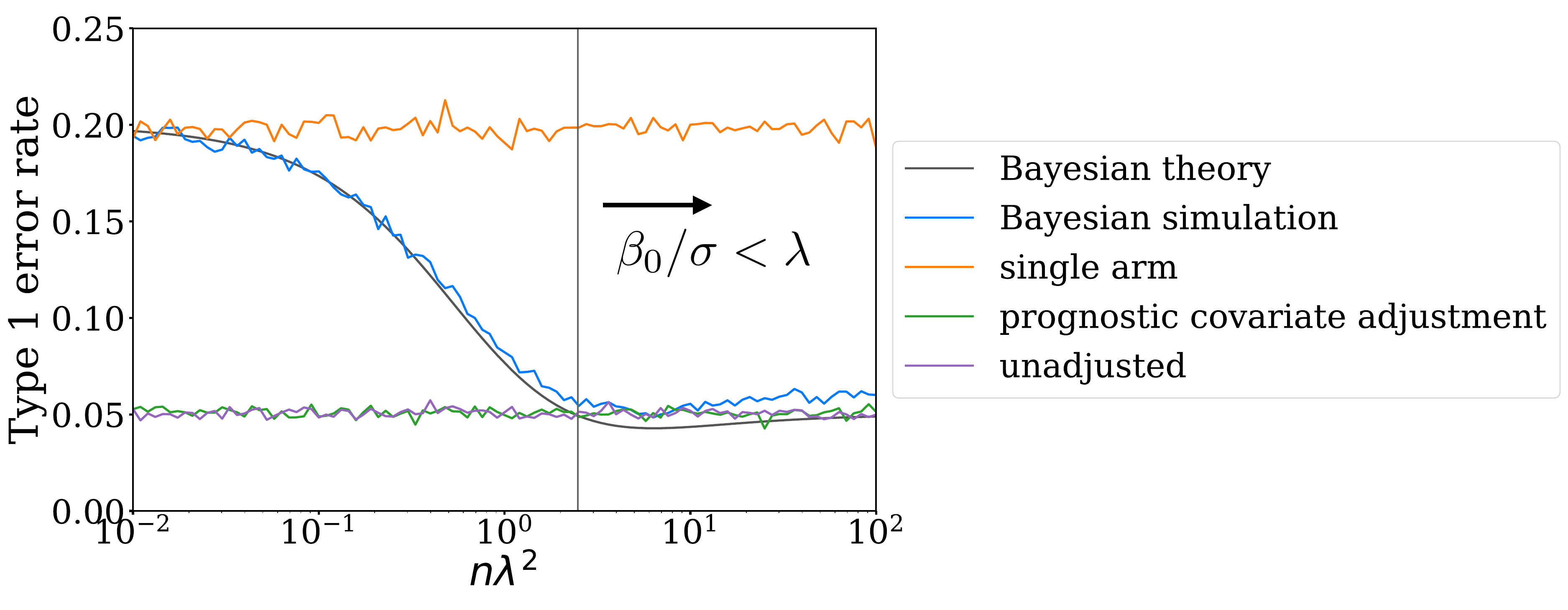}
\end{center}
\caption{{\bf The type I error rate of Bayesian prognostic covariate adjustment moves between two simple limits as the prior becomes narrow or flat.} Type I error rate as a function of $n\lambda^2$ for fixed $n = 1000$.  The value of $\beta_0$ is chosen so that the single arm estimator has a type I error rate of 20\%; the vertical line denotes the point where $\lambda = \beta_0 / \sigma$, above which the theory predicts type I error is at most $\alpha$.  Type I error from the Bayesian estimator is shown along with power from the single-arm, prognostic covariate adjustment, and unadjusted estimators.
\label{fig:gridA}}
\end{figure}

In \Fig{gridA}, the type I error rate computed from simulation is compared to the theoretical predictions, with only small deviations observed.  For fixed $\beta_0$, the type I error rate increases as $n \lambda^2$ decreases, and for $n \lambda^2 \not\ll 1$ it is significantly less than that of the single arm estimator.

Additionally, we performed simulations for various $n$ and $n\lambda^2$ and found generally good agreement between theoretical predictions and simulations, with the largest corrections coming at small $n$ or large $\lambda$ (see Appendix~\ref{sec:visual}).  This confirms that the type I error rate is bounded by the significance $\alpha$ at large $n$ and small $\lambda$.  Furthermore, as long as $n \lambda^2$ isn't too small, the type I error dependence on $\beta_0 / \sigma$ is much smaller than the single-arm estimator.  Note that the single-arm estimator is independent of the prior, which is a component of Bayesian prognostic covariate adjustment.

\subsubsection{Power}
\label{sec:simpower}

The power of the Bayesian estimator is predicted to depend on $\lambda$, ranging between the prognostic covariate adjustment limit when $\lambda$ is large ($n \lambda^2 \gg 1$) to a single-arm estimator when $\lambda$ is small ($n \lambda^2 \ll 1$). In \Fig{gridE}, we fix $\beta_1$ and compute the empirical power through simulations.  The power is shown as the precision of the prior varies between $n \lambda^2 \ll 1$ and $n \lambda^2 \gg 1$.   The Bayesian results are compared to three frequentist estimators: the single-arm estimator and estimators with and without prognostic covariate adjustment.  The agreement between the theoretical power and the simulations is good in all comparisons.  The power interoplates smoothly between the single-arm estimator when $n \lambda^2 \ll 1$ and the prognostic covariate adjustment estimator when $n \lambda^2 \gg 1$.

\begin{figure}[htp!]
\begin{center}
\includegraphics[width=5.5in]{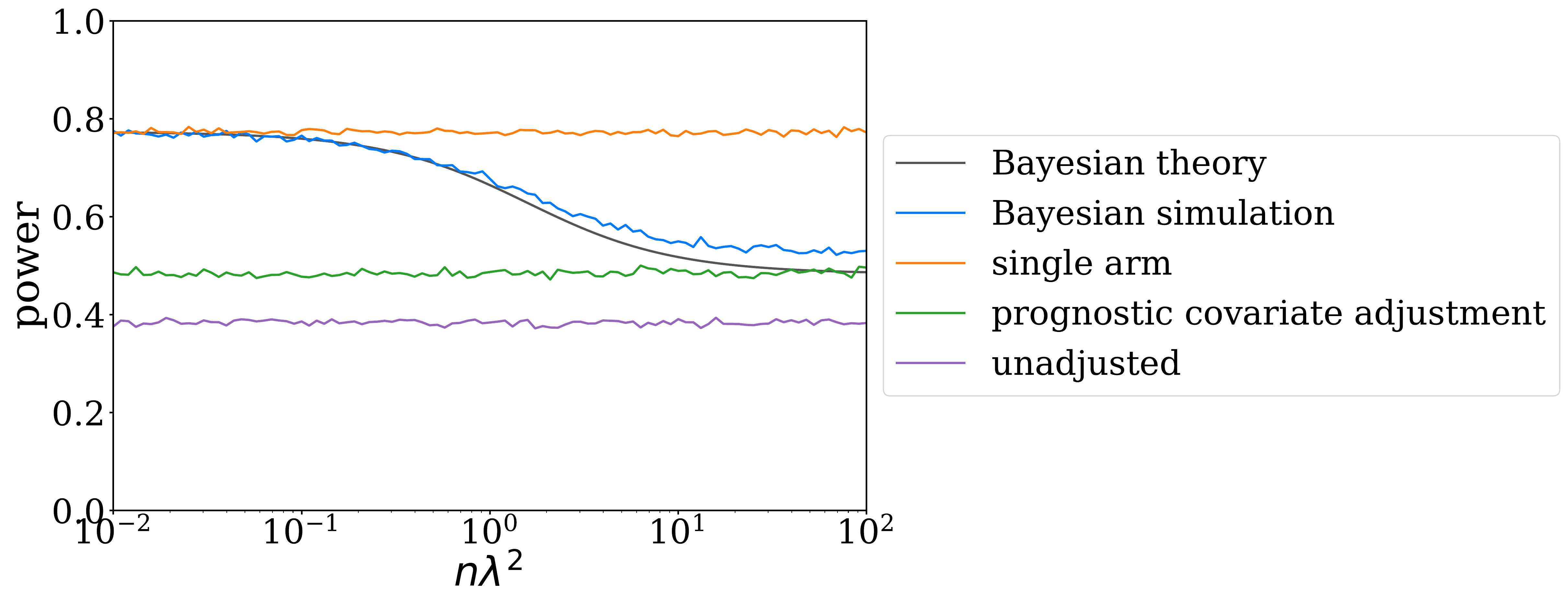}
\end{center}
\caption{{\bf The power of Bayesian prognostic covariate adjustment moves between two simple limits as the prior becomes narrow or flat.} Power as a function of $n \lambda^2$.  The value of $\beta_1$ is chosen so that the prognostic covariate adjustment estimator has a power of 50\%.  Power from the Bayesian estimator is shown along with power from various frequentist estimators: the single-arm estimator and estimators with and without prognostic covariate adjustment.
\label{fig:gridE}}
\end{figure}

\subsection{Analysis of a past clinical trial}

To study the empirical value of the analysis method presented in this work, we employ it in the analysis of a past Alzheimer's Disease (AD) clinical trial.  AD is an ideal setting to apply this method, as the complexity of disease progression and the abundance of historical data from clinical trials and observational studies allow for the creation of prognostic scores.

\subsubsection{Choosing the prior}
\label{sec:lambda}

The width of the prior determines the operating characteristics of the estimator, and a robust prior can allow for type I error rate control while still allowing power gains over prognostic covariate adjustment.

In Section \ref{sec:hist} we noted that learning a value for $\lambda$ is straightforward when there is a large quantity of historical data. In practice, however, there may be few relevant historical trials or the sample sizes for any historical trials may be small. We discuss two possible methods to estimate the prior in that case, one using subject-level data and the other using study-level data.  Study-level approaches are useful if the prognostic score bias ($\beta_0$) is expected to vary substantially between studies.  Subject-level approaches are useful when studies only partially overlap with the trial of interest, or if fewer studies are available.

To estimate a subject-level prior, we pool all relevant historical trials and we extract pairs ($Y_i$, $M_i$) of observed and predicted outcomes for each control subject $i$ in the historical data set.  Let $N$ be the total number of historical subjects, and define $E_{\rm all} = \hat{\beta}_0/\hat{\sigma}$, where:
\begin{align}
    \hat{\beta}_0 &= \frac{1}{N} \sum_i (Y_i - M_i)\\
    \hat{\sigma}^2 &= \frac{1}{N} \sum_i (Y_i - M_i - \hat{\beta}_0)^2
\end{align}
If the distributions of pairs $(Y_i, M_i)$ in all historical trials are similar to that of control subjects in the target trial, and if $N$ is large, $E_{all}$ should be a good estimate of the target ratio: $\beta_0/\sigma$. However, if $N$ is small enough that $1/\sqrt{N}$ is large compared with the true magnitude of the target ratio, there is a non-trivial probability that $|E_{all}|$ will underestimate that magnitude by a large relative margin. For that reason, we take $\lambda$ as $\max(\frac{3}{\sqrt{N}}, |E_{\rm all}|)$.

To estimate a study-level prior, we extract pairs ($Y_{i,j}$, $M_{i,j}$) for each historical trial, $j = 1, \dots, m$. We define $N_j$ as the sample size for each, and we set $E_j = \hat{\beta}_{0,j}/\hat{\sigma}_j$, where:
\begin{align}
    \hat{\beta}_{0,j} &= \frac{1}{N_j} \sum_i (Y_{i,j} - M_{i,j})\\
    \hat{\sigma}_j^2 &= \frac{1}{N_j} \sum_i (Y_{i,j} - M_{i,j} - \hat{\beta}_{0,j})^2
\end{align}

Since we will assume a $N(0, \lambda^2)$ prior for the target ratio (see eq. \ref{eq:beta_eps}), it is natural to model these $E_j$ as IID draws from some normal distribution with mean 0. The variance of that normal distribution is treated as an unknown parameter, and we take $\lambda^2$ as the upper limit of a 95\% maximum likelihood confidence interval for that variance. Thus:
\begin{equation}
    \label{eq:lambda_study} \lambda = \sqrt{\frac{\sum_j E_j^2}{(F^{\chi^2}_m)^{-1}(0.025)}},
\end{equation}
where $F^{\chi^2}_m$ denotes the distribution function of a $\chi^2_m$ random variable.

These two priors are compared in the example below.

\subsubsection{Reanalysis}

To facilitate comparison with prognostic covariate adjustment, we use the same clinical trial data and prognostic model as in Schuler et al.~\cite{freq-methods}.  The prognostic model predicts the progression from baseline in the ADAS-Cog 11 score, a cognitive assessment that is a standard primary endpoint in AD clinical trials for which higher scores indicate increased disease severity.  The prognostic model is a random forest trained on data from 6,919 patients across the AD spectrum. These data were provided by the Alzheimer's Disease Neuroimaging Initiative (ADNI) and the Critical Path for Alzheimer's Disease (CPAD) consortium \cite{cpad1,cpad2}, and included measurements of ADAS-Cog 11 at 6-month or more frequent intervals post-baseline. The ADNI dataset is made up of longitudinal data from 4 sequential large observational studies in Alzheimer's disease, while the CPAD dataset is made up of control arm data from 29 Alzheimer's disease clinical trials. These data also included a common set of 37 baseline covariates, which were imputed to a column mean where missing.  

The clinical trial analyzed was originally conducted to determine if docosahexaenoic acid (DHA) supplementation slows cognitive and functional decline for individuals with mild to moderate Alzheimer's disease \cite{dha}. The study was performed through the Alzheimer's Disease Cooperative Study (ADCS), a consortium of academic medical centers and private Alzheimer disease clinics funded by the National Institute on Aging to conduct clinical trials on Alzheimer disease. A total of 238 subjects in the trial were treated with DHA while 164 were given placebo, and ADAS-Cog 11 was a co-primary endpoint of the 18-month study.  The same covariates used to predict the prognostic score were also measured in this trial.  As in the data used to define the prognostic score, any missing covariate values were mean-imputed in our analysis.

To perform the Bayesian analysis of the DHA trial, a prior on the bias of the prognostic model is needed.  We carry out the two methods suggested in Section~\ref{sec:lambda}, setting the width of the prior using either subject-level or study-level methods.  As discussed, study-level methods can account for inter-study variability and bias, but are limited by the availability of appropriate studies.  In this case, 3 studies have sufficiently similar disease severity to the DHA trial and ADAS-Cog 11 measurements matching the trial duration.  These studies have at total of 1689 enrolled placebo subjects, 1138 of whom have data relevant for determining the prior.  The subject-level approach uses 1217 subjects with similar disease severity to the DHA trial and ADAS-Cog 11 measurements matching the trial duration.  The result of these approaches are priors with standard deviation $\lambda$ of 0.17 (study-level) and 0.029 (subject-level).  As the study has 272 subjects that complete treatment, this amounts to a value of $n \lambda^2$ of 7.5 for the study-level prior and 0.22 for the subject-level prior.

\begin{table}[h!]
\label{tab:results}
\centering
\begin{tabular}{ |c|c| } 
\hline
Analysis & Result \\
\hline
\hline
Unadjusted & -0.10 $\pm$ 2.02 \\
\hline
Prognostic Covariate Adjustment & 0.22 $\pm$ 1.82 \\
\hline 
Bayesian Prognostic Covariate Adjustment (study-level prior) & 0.48 $\pm$ 1.63 \\
\hline 
Bayesian Prognostic Covariate Adjustment (subject-level prior) & 0.79 $\pm$ 1.49 \\
\hline
Single-arm & 1.29 $\pm$ 1.25 \\
\hline 
\end{tabular}
\caption{Results from the reanalysis of the DHA trial. Results are shown in terms of estimated effect $\pm$ 1.96 $\times$ estimated standard deviation.  For the Bayesian estimates, the standard deviation of the posterior probability distribution is used.  Bayesian results are shown for the two choices of priors discussed in the text.}
\end{table}

The results of the analysis of the DHA trial are shown in Table~\ref{tab:results}. Prognostic covariate adjustment with strict type I error rate control provides a substantial reduction in uncertainty compared to an unadjusted analysis. In similar fashion, Bayesian prognostic covariate adjustment with either prior provides substantial reductions in uncertainty compared to the frequentist formulation of prognostic covariate adjustment. None of these analyses, however, suggest that the observed effect is statistically significant, in keeping with the conclusions of the original study~\cite{dha}. By contrast, a single-arm analysis that uses the prognostic score in place of placebo arm outcomes shows an even larger reduction in uncertainty, leading to the observation of a statistically significant effect (in the wrong direction, suggestion the treatment leads to worsening of the disease). This illustrates the potential for bias in single-arm studies using prognostic models to replace concurrent controls arm, which is effectively mitigated using Bayesian prognostic covariate adjustment with a prior on the model bias.

\section{Discussion}
\label{sec:discussion}

In this paper, we have presented a novel method for analyzing clinical trials that combines the advantages of prognostic covariate adjustment with the opportunity to leverage additional external information through a Bayesian formulation to further increase power while maintaining approximate type I error rate control. Given the breadth of patient data becoming available that can be used to build and assess prognostic models, this methodology has the potential to radically improve the efficiency of clinical trials across many disease areas.

Although Bayesian methods are becoming more popular in clinical trials, insistence on ``reasonable control of type I error rates'' by regulatory authorities often limits their use, and typically necessitates extensive simulations to characterize the operating characteristics for each new target trial. Here we showed that Bayesian prognostic covariate adjustment provides reasonable control of the type I error rate under mild assumptions about the prior distribution. Moreover, we provided analytic expressions for the type I error rate and power; technically, the expressions hold in an asymptotic regime, but as we demonstrated with simulations, the formulae are generally accurate across the parameter spectrum.

The Bayesian estimator defined here utilizes a prognostic score for control outcomes as well as a prior on the same outcome.  The operating characteristics of the Bayesian estimator is shown to interpolate between the operating characteristics of the prognostic covariate adjustment method presented by Schuler et al. \cite{freq-methods} and a single-arm trial that treats the prognostic scores as actual placebo outcomes. Not surprisingly, as the prior on the predicted control outcome from the prognostic score becomes very wide, the estimator behaves like prognostic covariate adjustment and, as it becomes increasingly narrow, the estimator behaves more and more like a single-arm trial.  Within the limits of the data used to define and place a prior on the prognostic control outcome, this approach allows practitioners across the risk spectrum for type I error control to utilize this method.

Both prognostic covariate adjustment and its Bayesian counterpart presented here benefit from better prognostic models. That is, the increase in power relative to an unadjusted analysis depends on the correlation of the prognostic score with the outcomes and, in the case of the Bayesian approach, the average bias in the prognostic score on the control group. Therefore, machine learning-based predictive models that are able to learn nonlinear relationships from large historical databases may be particularly useful when coupled with these prognostic covariate adjustment methods.

An important set of assumptions underlies this work; in particular, that the outcomes and prognostic scores are linearly related, the treatment effect is constant, and errors are homoscedastic and normally distributed. We expect that approximate linearity should hold for any reasonable prognostic model that has been engineered to predict placebo outcomes. However, additional research will be required to understand (and perhaps improve) the robustness of Bayesian prognostic covariate adjustment to violations of the other assumptions. Regarding outcome measures that are not continuous, such as binary or time-to-event outcomes, we speculate that generalized linear models could relate such outcomes to prognostic scores, offering a solution that is qualitatively similar to the method of this paper.

Finally there is the question of data availability. The type and quantity of historical data impacts how the prior may be obtained in practice. This paper highlights how study-level methods can produce highly variable values for $\lambda$ when the historical sample sizes are small. On the other hand, subject-level methods are less noisy but fail to capture heterogeneity between historical studies. We hope that future work will unify these methods, identifying reliable methods of prior selection. Besides issues with the historical data, there could be missing measurements for the target trial. For instance, it may not be possible to record for every subject all of the baseline features that the prognostic model requires. Alternatively, subjects may drop out over the course of the target trial, leading to missing outcome data. In either case, it is interesting to consider how our approach should be combined with imputation methods to resolve these issues.

\section{Data Availability}

Certain data used in the preparation of this article were obtained from the Alzheimer’s Disease Neuroimaging Initiative (ADNI) database (\href{url}{adni.loni.usc.edu}). The ADNI was launched in 2003 as a public-private partnership, led by Principal Investigator Michael W. Weiner, MD. The primary goal of ADNI has been to test whether serial magnetic resonance imaging (MRI), positron emission tomography (PET), other biological markers, and clinical and neuropsychological assessment can be combined to measure the progression of mild cognitive impairment (MCI) and early Alzheimer’s disease (AD). For up-to-date information, see \href{url}{www.adni-info.org}.

Certain data used in the preparation of this article were obtained from the Critical Path for Alzheimer's Disease (CPAD) database. In 2008, Critical Path Institute, in collaboration with the Engelberg Center for Health Care Reform at the Brookings Institution, formed the Coalition Against Major Diseases (CAMD), which was then renamed to CPAD in 2018. The Coalition brings together patient groups, biopharmaceutical companies, and scientists from academia, the U.S. Food and Drug Administration (FDA), the European Medicines Agency (EMA), the National Institute of Neurological Disorders and Stroke (NINDS), and the National Institute on Aging (NIA). CPAD currently includes over 200 scientists, drug development and regulatory agency professionals, from member and non-member organizations. The data available in the CPAD database has been volunteered by CPAD member companies and non-member organizations.

Certain data used in the preparation of this article were obtained from the University of California, San Diego Alzheimer’s Disease Cooperative Study Legacy database.

\section{Acknowledgments}

Data collection and sharing for this project was funded in part by the Alzheimer's Disease Neuroimaging Initiative (ADNI) (National Institutes of Health Grant U01 AG024904) and DOD ADNI (Department of Defense award number W81XWH-12-2-0012). ADNI is funded by the National Institute on Aging, the National Institute of Biomedical Imaging and Bioengineering, and through generous contributions from the following: AbbVie, Alzheimer’s Association; Alzheimer’s Drug Discovery Foundation; Araclon Biotech; BioClinica, Inc.; Biogen; Bristol-Myers Squibb Company; CereSpir, Inc.; Cogstate; Eisai Inc.; Elan Pharmaceuticals, Inc.; Eli Lilly and Company; EuroImmun; F. Hoffmann-La Roche Ltd and its affiliated company Genentech, Inc.; Fujirebio; GE Healthcare; IXICO Ltd.; Janssen Alzheimer Immunotherapy Research \& Development, LLC.; Johnson \& Johnson Pharmaceutical Research \& Development LLC.; Lumosity; Lundbeck; Merck \& Co., Inc.; Meso Scale Diagnostics, LLC.; NeuroRx Research; Neurotrack Technologies; Novartis Pharmaceuticals Corporation; Pfizer Inc.; Piramal Imaging; Servier; Takeda Pharmaceutical Company; and Transition Therapeutics. The Canadian Institutes of Health Research is providing funds to support ADNI clinical sites in Canada. Private sector contributions are facilitated by the Foundation for the National Institutes of Health (\href{url}{www.fnih.org}). The grantee organization is the Northern California Institute for Research and Education, and the study is coordinated by the Alzheimer’s Therapeutic Research Institute at the University of Southern California. ADNI data are disseminated by the Laboratory for Neuro Imaging at the University of Southern California.

Data collection and sharing for this project was funded in part by the University of California, San Diego Alzheimer’s Disease Cooperative Study (ADCS) (National Institute on Aging Grant Number U19AG010483).

\bibliographystyle{plain}
\bibliography{bayes}

\appendix

\section{Proofs of technical results}
\label{sec:proofs}

\begin{proof}[Proof of Lemma \ref{lemma:single_power}]
The single-arm analysis compares the following t-statistic against the distribution function of a $t_{pn-1}$ random variable:
\begin{equation}
    \label{eq:t_stat_single}
    T = \frac{\bar{Y}_t - \bar{M}_t}{S_t/\sqrt{pn - 1}}
\end{equation}
where $Y_t$ and $M_t$ are the average values of $Y$ and $M$ on the active treatment arm, and $S_t^2$ denotes the sample variance of $(Y-M)$ on the active treatment arm. The asymptotic result follows when we note that the distribution of $t_{pn - 1}$ approaches a standard Gaussian, and the distribution of $T$ approaches $N((\beta_0 + \beta_1)\sqrt{pn}/\sigma, 1)$.
\end{proof}

\begin{proof}[Proof of Lemma \ref{lemma:post}]
Let
\[ \tilde{\beta} = \begin{pmatrix}
1 & p & 0\\
0 & 1 & 0\\
0 & 0 & 1
\end{pmatrix} \begin{pmatrix}
\beta_0\\
\beta_1\\
\beta_2
\end{pmatrix}. \]
Then (conditional on all prognostic scores and treatment assignments) \eqref{eq:lin_mod} is equivalent to the linear model,
\[ \mathbf{Y} \sim \mathbf{X} \tilde{\beta} + N(0, \sigma^2 I_n), \]
where $\mathbf{Y}$ and $\mathbf{X}$ were defined in the statement of this lemma. Under this parameterization, $\mathcal{P}_\epsilon$ becomes:
\begin{align*}
    \sigma^2 &\sim Inverse-Gamma(\epsilon, \epsilon)\\
    \tilde{\beta} | \sigma^2 &\sim N \left (0, \sigma^2 V_{0,\epsilon} \right )
\end{align*}
where
\[ V_{0, \epsilon} = \begin{pmatrix}
1 & p & 0\\
0 & 1 & 0\\
0 & 0 & 1
\end{pmatrix} \begin{pmatrix}
\lambda^2 & 0 & 0\\
0 & 1/\epsilon & 0\\
0 & 0 & 1/\epsilon
\end{pmatrix} \begin{pmatrix}
1 & p & 0\\
0 & 1 & 0\\
0 & 0 & 1
\end{pmatrix}^T = \begin{pmatrix}
\lambda^2 + p^2/\epsilon & p/\epsilon & 0\\
p/\epsilon & 1/\epsilon & 0\\
0 & 0 & 1/\epsilon
\end{pmatrix}\]

Standard theory on Bayesian linear regression with conjugate priors establishes the posterior distribution $\mathcal{P}'_\epsilon$ as:
\begin{align*}
    \sigma^2 &\sim Inverse-Gamma(n/2 + \epsilon, S_\epsilon^2/2 + \epsilon)\\
    \tilde{\beta} | \sigma^2 &\sim N \left (\mu_\epsilon, \sigma^2 V_\epsilon \right )
\end{align*}
where
\begin{align*}
    V^{-1}_\epsilon &= V_{0, \epsilon}^{-1} + \mathbf{X}^T \mathbf{X} = \begin{pmatrix} 1/\lambda^2 & -p/\lambda^2 & 0\\
    -p/\lambda^2 & p^2/\lambda^2 + \epsilon & 0\\
    0 & 0 & \epsilon
    \end{pmatrix} + \mathbf{X}^T \mathbf{X},\\
    \mu_\epsilon &= V_\epsilon \mathbf{X}^T \mathbf{Y}, \\
    S^2_\epsilon &= \mathbf{Y}^T \mathbf{Y} - \mu_\epsilon^T V_\epsilon^{-1} \mu_\epsilon.
\end{align*}

On the almost sure event that $V^{-1}$ (as defined in the statement of this lemma) has an inverse, $V$, this implies that $V_\epsilon \to V$. Hence on that event, $\mu_\epsilon \to \mu$ and $S_\epsilon^2 \to S^2$ also. The almost sure convergence of these variables is sufficient to ensure that $\mathcal{P}'_\epsilon \to \mathcal{P}'$.

The marginal distribution of $\beta_1$ under $\mathcal{P}'$ follows from standard properties of Normal Inverse-Gamma distributions.
\end{proof}

\begin{proof}[Proof of Theorem \ref{thm:rejection_rate}]
\begin{enumerate}
    \item Starting from \eqref{eq:V}, some algebra gives:
\[ V^{-1} = \begin{pmatrix}
n + 1/\lambda^2 & -p/\lambda^2 & 0\\
-p/\lambda^2 & np(1-p) + p^2/\lambda^2 & np(1-p) (\bar{M}_t - \bar{M}_c)\\
0 & np(1-p) (\bar{M}_t - \bar{M}_c) & n S_M^2
\end{pmatrix}.
\]
where $\bar{M}_t$ and $\bar{M}_c$ are the average values of $M$ observed in the treatment and control groups respectively, and $S_M^2$ is the sample variance of $M$ across all subjects. After inversion and using the fact that $\bar{M}_t - \bar{M}_c = O_P(n^{-1/2})$, we find that the $\beta_1$ row of $V$ is given by:
\[ V_{1 \cdot} = \frac{1}{n} \begin{pmatrix}
    \frac{1}{n\lambda^2(1-p) + 1} & \frac{n\lambda^2 + 1}{n\lambda^2p(1-p) + p} + O_P(n^{-1/2}) & \frac{(\bar{M}_c - \bar{M}_t)(2n\lambda^2(1-p) + 1)}{S_M^2(2n\lambda^2(1-p) + 2)}
\end{pmatrix} \]
\item From Lemma \ref{lemma:post}, we have the following posterior distribution for $\beta_1$:
\[ \mu_1 + \sqrt{V_{11} S^2 / n} \,\, t_n = N(\mu_1, V_{11} S^2 / n) + o(1) \]
Thus, in the given asymptotic, the rejection event coincides almost surely with the event that $|\mu_1| \geq \Phi(1 - \alpha/2) \sqrt{V_{11} S^2 / n}$. To derive the rejection rate, we need to unpack the distributions of $\mu_1$ and $S^2$.

Firstly from \eqref{eq:mu}:
\begin{align*}
    \mu_1 &= (n V_{1 \cdot}) \left (\frac{1}{n} \mathbf{X}^T\mathbf{Y}\right)\\
    &= \begin{pmatrix}
    \frac{1}{n\lambda^2(1-p) + 1} & \frac{n\lambda^2 + 1}{n\lambda^2p(1-p) + p} & \frac{(\bar{M}_c - \bar{M}_t)(2n\lambda^2(1-p) + 1)}{S_M^2(2n\lambda^2(1-p) + 2)}
    \end{pmatrix} \begin{pmatrix}
    (1-p) \bar{Y}_c + p \bar{Y}_t - \bar{M}\\
    p(1-p)(\bar{Y}_t - \bar{Y}_c)\\
    \bar{MY} - \bar{M}\bar{Y}
    \end{pmatrix} + o_P(n^{-1/2})\\
    &= \left \{ \bar{Y}_t - \bar{M} - \hat{\beta}_2(\bar{M}_t - \bar{M}) \right \} - \left ( \frac{n\lambda^2(1 - p)}{n\lambda^2(1-p) + 1} \right ) \left \{ \bar{Y}_c - \bar{M} - \hat{\beta}_2 (\bar{M}_c - \bar{M}) \right \} + o_P(n^{-1/2})
\end{align*}
Here $\bar{Y}_c$ and $\bar{Y}_t$ denote the mean responses observed in the control and treatment groups respectively, $\bar{MY}$ is the mean value of $M \times Y$ across all patients in the trial, and
\begin{equation}
    \hat{\beta}_2 = \frac{ \bar{MY} - \bar{M} \bar{Y} }{S_M^2}.
\end{equation}

After fixing values for $\beta_0, \beta_1, \beta_2$ and $\sigma^2$, we observe that by the weak Law of Large Numbers,
\begin{align*}
    \hat{\beta}_2 &= \frac{1}{S_M^2} \left \{ E(\bar{MY}|M_1, M_2, \dots) - \bar{M} E(Y_1|M_1, M_2, \dots) \right \} + o_P(1)\\
    &= \frac{1}{S_M^2} \left \{ [\beta_0 \bar{M} + p \beta_1 \bar{M}_t + \beta_2 S_M^2 + \bar{M}^2] - \bar{M} [\beta_0 + p \beta_1 + \bar{M}] \right \} + o_P(1)\\
    &= \beta_2 + \frac{p\beta_1(\bar{M}_t - \bar{M})}{S_M^2} + o_P(1)\\
    &= \beta_2 + o_P(1)
\end{align*}
Thus we may replace $\hat{\beta}_2$ by $\beta_2$ in the above expression, while still maintaining a result that is correct up to $o_P(n^{-1/2})$.
\begin{align*}
    \bar{Y}_c - \bar{M} - \beta_2(\bar{M}_c - \bar{M}) &\sim N \left ( \beta_0, \frac{\sigma^2}{(1-p)n} \right )\\
    \bar{Y}_t - \bar{M} - \beta_2(\bar{M}_t - \bar{M}) &\sim N \left ( \beta_0 + \beta_1, \frac{\sigma^2}{pn} \right ).
\end{align*}
and so
\[ \mu_1 = N(\tau, \sigma^2 \hat{V}) + o_P(1) \]

Secondly from \eqref{eq:S^2}:
\begin{align*}
    \frac{1}{n} S^2 &= \frac{1}{n} (\mathbf{Y}^T \mathbf{Y} - \mu^T V^{-1} \mu)\\
    &= \frac{1}{n} \left (\mathbf{Y}^T \mathbf{Y} - (V \mathbf{X} \mathbf{Y})^T \mathbf{X}^T \mathbf{Y} \right)\\
    &= \frac{1}{n} || \mathbf{Y} - \mathbf{X} (\mathbf{X}^T \mathbf{X})^{-1} \mathbf{X}^T \mathbf{Y} ||^2 + \frac{1}{n} \mathbf{Y}^T \mathbf{X} [ (\mathbf{X}^T \mathbf{X})^{-1} - V ] \mathbf{X}^T \mathbf{Y}
\end{align*}
We recognize the first term as the maximum likelihood estimate for estimating $\sigma^2$, which is commonly known to be consistent. For the second term, some algebra gives:
\begin{align*}
    n [ (\mathbf{X}^T \mathbf{X})^{-1} - V ] &= \begin{pmatrix}
    1 & 0 & 0\\
    0 & 1/p(1-p) & 0\\
    0 & 0 & 1/S_M^2
    \end{pmatrix} - \begin{pmatrix}
    \frac{n \lambda^2(1-p) + p}{n \lambda^2(1-p) + 1} & \frac{1}{n \lambda^2(1-p) + 1} & 0\\
    \frac{1}{n \lambda^2 (1-p) + 1} & \frac{n \lambda^2 + 1}{n \lambda^2 p(1-p) + p} & 0\\
    0 & 0 & 1/S_M^2
    \end{pmatrix} + o_P(1)\\
    &= \frac{1}{n\lambda^2(1 - p) + 1}\begin{pmatrix}
    1-p & -1 & 0\\
    -1 & 1/(1-p) & 0\\
    0 & 0 & 0
    \end{pmatrix} + o_P(1)
\end{align*}
Thus, up to terms that are $o_P(1)$:
\begin{align*}
\frac{1}{n} S^2 &= \sigma^2 + \frac{1}{n\lambda^2(1-p) + 1} \begin{pmatrix}
    (1-p)\bar{Y}_c + p\bar{Y}_t - \bar{M}\\
    p(1-p)(\bar{Y}_t - \bar{Y}_c)
    \end{pmatrix}^T \begin{pmatrix}
    1-p & -1\\
    -1 & 1/(1-p)
    \end{pmatrix} \begin{pmatrix}
    (1-p)\bar{Y}_c + p\bar{Y}_t - \bar{M}\\
    p(1-p)(\bar{Y}_t - \bar{Y}_c)
    \end{pmatrix}\\
    &= \sigma^2 + \frac{1}{n\lambda^2(1-p) + 1} \begin{pmatrix}
    (1-p)\bar{Y}_c + p\bar{Y}_t - \bar{M}\\
    p(1-p)(\bar{Y}_t - \bar{Y}_c)
    \end{pmatrix}^T \begin{pmatrix}
    (1-p)(\bar{Y}_c - \bar{M})\\
    -(\bar{Y}_c - \bar{M})
    \end{pmatrix}\\
    &= \sigma^2 + \frac{(1-p)(\bar{Y}_c - \bar{M})^2}{n\lambda^2(1-p) + 1}\\
    &= \sigma^2 + \frac{(1-p) \beta_0^2}{n\lambda^2(1-p) + 1}
\end{align*}

Altogether the asymptotic rejection rate is given by the probability that
\[ |\sigma \sqrt{\hat{V}}Z + \tau| > \Phi^{-1}(1 - \alpha/2) \sqrt{V_{11} \left (\sigma^2 + \frac{(1-p)\beta_0^2}{n\lambda^2(1-p) + 1} \right )} \]
when $Z$ has a standard normal distribution. Noting this that inequality may be rewritten as
\begin{align*}
    |Z + \frac{\tau}{\sigma\sqrt{\hat{V}}}| &> \Phi^{-1}(1 - \alpha/2) \sqrt{ \frac{V_{11}}{\hat{V}} \left (1 + \frac{(1-p)\beta_0^2}{\sigma^2 (n\lambda^2(1-p) + 1)} \right ) },
\end{align*}
we see that this is equal to the rate given in the theorem.
\end{enumerate}
\end{proof}

\section{Prior on the slope parameter}
\label{sec:beta2}

In Section \ref{sec:methods}, we justified a prior which only incorporates information about $\beta_0$ on the grounds that prior information about $\beta_2$ does not offer substantial power gains. Here we formalize that claim.

Consider the following family of Normal Inverse-Gamma priors, $\mathcal{P}_\epsilon^{\lambda_2}$, where the extra variance parameter, $\lambda_2^2$, is used to capture historical information about $\beta_2$. We will assume that all of the priors are ``correct'', in the sense that the prior mean, $\mu_2^0$, coincides with the true value of $\beta_2$.
\begin{align}
    \sigma^2 &\sim Inverse-Gamma(\epsilon, \epsilon)\\
    \begin{pmatrix}
    \beta_0\\
    \beta_1\\
    \beta_2
    \end{pmatrix} | \sigma^2 &\sim N \left ( \begin{pmatrix}
    0\\
    0\\
    \mu_2^0
    \end{pmatrix}, \sigma^2 \begin{pmatrix} \lambda^2 & 0 & 0\\
    0 & 1/\epsilon & 0\\
    0 & 0 & \lambda_2^2
    \end{pmatrix} \right )
\end{align}

Starting from $\mathcal{P}^{\lambda_2}_\epsilon$, let $(\mathcal{P}_\epsilon^{\lambda_2})'$ denote the joint posterior for all parameters. By the same argument as in Lemma \ref{lemma:post}, we see that these posterior distributions converge almost surely to the following random distribution $(\mathcal{P}^{\lambda_2})'$ as $\epsilon \to 0$:
\begin{align}
    \label{eq:sigma_post2}
    \sigma^2 &\sim Inverse-Gamma(n/2, (S^{\lambda_2})^2/2)\\
    \label{eq:beta_post2}
    \begin{pmatrix}
    \beta_0 + p \beta_1\\
    \beta_1\\
    \beta_2
    \end{pmatrix} | \, \sigma^2 &\sim N(\mu^{\lambda_2}, \sigma^2 V^{\lambda_2})
\end{align}
where
\begin{align}
    \label{eq:V_2}
    (V^{\lambda_2})^{-1} &= \begin{pmatrix} 1/\lambda^2 & -p/\lambda^2 & 0\\
    -p/\lambda^2 & p^2/\lambda^2 & 0\\
    0 & 0 & 1/\lambda^2_2
    \end{pmatrix} + \mathbf{X}^T \mathbf{X},\\
    \label{eq:mu_2}
    \mu^{\lambda_2} &= V^{\lambda_2} \left \{ \mathbf{X}^T \mathbf{Y} + \begin{pmatrix} 0\\
    0\\
   \mu_2^0/\lambda^2_2
    \end{pmatrix} \right \}\\
    \label{eq:S_2^2}
    (S^{\lambda_2})^2 &= \mathbf{Y}^T \mathbf{Y} - (\mu^{\lambda_2})^T (V^{\lambda_2})^{-1} \mu^{\lambda_2} + (\mu_2^0)^2/\lambda_2^2.
\end{align}

\begin{theorem}
\label{thm:beta2}
Fix the true values of $\beta_0, \beta_1, \beta_2$ and $\sigma^2$, with $\beta_2 = \mu_2^0$. Consider the limit where $\lambda^2 \to 0$ and $n \to \infty$ such that $n\lambda^2 = O(1)$. The rejection rates associated with each $(\mathcal{P}^{\lambda_2})'$ asymptote to \eqref{eq:rejection_rate} uniformly over all $\lambda_2^2 > 0$.
\end{theorem}

\begin{proof}
By the same argument as in the proof of Theorem \ref{thm:rejection_rate}, the $\beta_1$ row of $V^{\lambda_2}$ is given by:
\[ V^{\lambda_2}_{1 \cdot} = \frac{1}{n} \begin{pmatrix}
    \frac{1}{n\lambda^2(1-p) + 1} & \frac{n\lambda^2 + 1}{n\lambda^2p(1-p) + p} + O_P(n^{-1/2}) & \frac{(\bar{M}_c - \bar{M}_t)(2n\lambda^2(1-p) + 1)}{(S_M^2 + 1/n\lambda_2^2)(2n\lambda^2(1-p) + 2)}
\end{pmatrix} \]
In particular, we see that $n(V^{\lambda_2}_{11} - V_{11}) \to_P 0$. Thus, in the given limit, the rejection event coincides almost surely with the event that $|\mu_1^{\lambda_2}| \geq \Phi(1 - \alpha/2) \sqrt{V_{11} (S^{\lambda_2})^2 / n}$.

From \eqref{eq:S_2^2}:
\begin{align*}
    (S^{\lambda_2})^2 &= \mathbf{Y}^T \mathbf{Y} - (\mu^{\lambda_2})^T (V^{\lambda_2})^{-1} \mu^{\lambda_2} + (\mu_2^0)^2/\lambda_2^2\\
    &= \mathbf{Y}^T \mathbf{Y} + \mu_2^0 (V^{\lambda_2})^{-1}_{2 \cdot} \mu^{\lambda_2} - \left \{ \mu^{\lambda_2} + \begin{pmatrix}
    0\\
    0\\
    \mu_2^0
    \end{pmatrix} \right \}^T (V^{\lambda_2})^{-1} \mu^{\lambda_2} + (\mu_2^0)^2/\lambda_2^2\\
    &= \mathbf{Y}^T \mathbf{Y} + \mu_2^0 \left \{ \begin{pmatrix}
    0\\
    0\\
    1/\lambda_2^2
    \end{pmatrix}^T + (V^{-1})_{2 \cdot} \right \} \mu^{\lambda_2} - \left \{ \mu^{\lambda_2} + \begin{pmatrix}
    0\\
    0\\
    \mu_2^0
    \end{pmatrix} \right \}^T \left \{ \mathbf{X}^T \mathbf{Y} + \begin{pmatrix} 0\\
    0\\
   \mu_2^0/\lambda^2_2
    \end{pmatrix} \right \} + (\mu_2^0)^2/\lambda_2^2\\
    &= \mathbf{Y}^T \mathbf{Y} + \mu_2^0 (V^{-1})_{2 \cdot} \mu^{\lambda_2} - \left \{ \mu^{\lambda_2} + \begin{pmatrix}
    0\\
    0\\
    \mu_2^0
    \end{pmatrix} \right \}^T \mathbf{X}^T \mathbf{Y}\\
    &= \mathbf{Y}^T \mathbf{Y} + \mu_2^0 (V^{-1})_{2 \cdot} \mu - \left \{ \mu + \begin{pmatrix}
    0\\
    0\\
    \mu_2^0
    \end{pmatrix} \right \}^T V^{-1} \mu + \left \{ \mu_2^0 (V^{-1})_{2 \cdot} - \mathbf{Y}^T \mathbf{X} \right \} (\mu^{\lambda_2} - \mu)\\
    &= S^2 + \left \{ \mu_2^0 (V^{-1})_{2 \cdot} - \mathbf{Y}^T \mathbf{X} \right \} (\mu^{\lambda_2} - \mu)
\end{align*}

We find that the second term here is uniformly $o_P(1)$ when $\beta_2 = \mu_2^0$, because:
\begin{align*}
    \mu^{\lambda_2} &= V^{\lambda_2} \left \{ \mathbf{X}^T \mathbf{Y} + \begin{pmatrix} 0\\
    0\\
   \mu_2^0/\lambda^2_2
    \end{pmatrix} \right \}\\
    &= V^{\lambda_2} \left \{ V^{-1} \mu + (V^{\lambda_2} - V^{-1}) \begin{pmatrix} \cdot\\
    \cdot\\
    \beta_2 \end{pmatrix} \right \}\\
    &= V^{\lambda_2} V^{-1} \mu + (I - V^{\lambda_2} V^{-1}) \begin{pmatrix} \mu_0\\
    \mu_1\\
    \beta_2 \end{pmatrix}
\end{align*}
where a dot indicates that that element of the vector is arbitrary. Letting $K^{\lambda_2}$ denote the largest eigenvalue of $(I - V^{\lambda_2} V^{-1})$, this implies: $|\mu^{\lambda_2} - \mu| \leq K^{\lambda_2} |\mu_2 - \beta_2| \leq |\mu_2 - \beta_2| = o_P(1)$.

Thus, in the given limit, the rejection event coincides almost surely with the event that $|\mu_1^{\lambda_2}| \geq \Phi(1 - \alpha/2) \sqrt{V_{11} S^2 / n}$. From \eqref{eq:mu_2}:
\begin{align*}
    \mu_1^{\lambda_2} &= (n V^{\lambda_2}_{1 \cdot}) \left \{ \frac{1}{n} \mathbf{X}^T \mathbf{Y} + \begin{pmatrix} 0\\
    0\\
   \mu_2^0/n\lambda^2_2
    \end{pmatrix} \right \}\\
    &= \begin{pmatrix}
    \frac{1}{n\lambda^2(1-p) + 1} & \frac{n\lambda^2 + 1}{n\lambda^2p(1-p) + p} & \frac{(\bar{M}_c - \bar{M}_t)(2n\lambda^2(1-p) + 1)}{(S_M^2 + 1/n\lambda_2^2)(2n\lambda^2(1-p) + 2)}
    \end{pmatrix} \begin{pmatrix}
    (1-p) \bar{Y}_c + p \bar{Y}_t - \bar{M}\\
    p(1-p)(\bar{Y}_t - \bar{Y}_c)\\
    S_M^2 \hat{\beta}_2 + \mu_2^0/n\lambda^2_2
    \end{pmatrix} + o_P(n^{-1/2})\\
    &= \left \{ \bar{Y}_t - \bar{M} - \hat{\beta}_2^{\lambda_2} (\bar{M}_t - \bar{M}) \right \} - \left ( \frac{n\lambda^2(1 - p)}{n\lambda^2(1-p) + 1} \right ) \left \{ \bar{Y}_c - \bar{M} - \hat{\beta}_2^{\lambda_2} (\bar{M}_c - \bar{M}) \right \} + o_P(n^{-1/2})
\end{align*}
where
\[ \hat{\beta}_2^{\lambda_2} = \frac{S_M^2 \hat{\beta}_2 + \mu_2^0/n\lambda^2_2}{S_M^2 + 1/n\lambda_2^2}. \]
After fixing values for $\beta_0, \beta_1, \beta_2$ and $\sigma^2$ with $\beta_2 = \mu_2^0$, we see that $|\hat{\beta}_2^{\lambda_2} - \beta_2| \leq |\hat{\beta}_2 - \beta_2| \to_P 0$. Hence, just as for Theorem \ref{thm:rejection_rate}, we find that $\mu_1^{\lambda_2} \sim N(\tau, \sigma^2 \hat{V})$, up to an error that is uniformly $o_P(1)$.

The remainder of this proof is identical to that of Theorem \ref{thm:rejection_rate}.

\end{proof}

\section{Visualizations of asymptotic theory}
\label{sec:visual}

This appendix provides some additional visualizations on the operating characteristics of our method. The figures here include both the asymptotic rejection rate given in Theorem \ref{thm:rejection_rate} and simulated rates, which have been computed using the approach described in Section \ref{sec:simulations}.

\begin{figure}[hbt!]
\begin{center}
\includegraphics[width=5.5in]{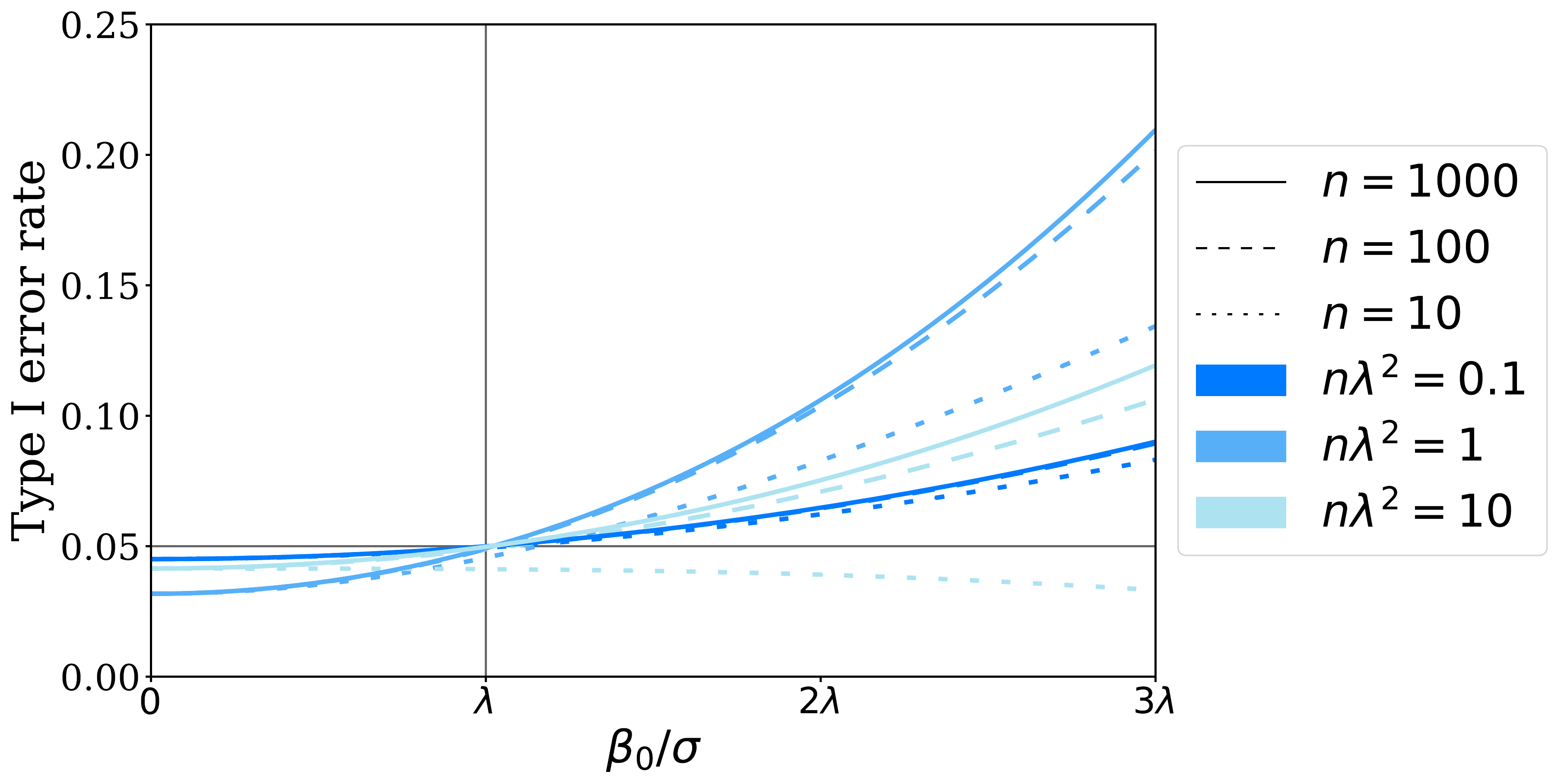}
\end{center}
\caption{{\bf The theoretical type I error rate is bounded by \boldsymbol{$\alpha$} when \boldsymbol{$\beta_0 / \sigma < \lambda$}}. The theoretical type I error rate as a function of $\beta_0 / \sigma$ in units of $\lambda$ for different values of $n$ and $n\lambda^2$, with $\alpha = 0.05$.  The line at $\beta_0 / \sigma = \lambda$ shows the region below which (for $\beta_0 / \sigma < \lambda$) the type I error rate is bounded above by $\alpha$.
\label{fig:type1_theory}}
\end{figure}

Figure~\ref{fig:type1_theory} shows that across a wide range of values for $n$ and $\lambda$, the condition $|\beta_0/\sigma| \leq \lambda$ ensures the asymptotic type I error rate is approximately at most $\alpha$.

In addition to the linear generative model discussed in the text, we also investigate type I error when the model assumptions do not hold. Here we use a generative model where the relationship between the prognostic score and outcomes is nonlinear, which is parameterized so that $\beta_0$ and $\sigma^2$ still represent the bias and residual variance of the prognostic model respectively:
\begin{align}
    M &\sim {\cal N} (0, 1) \,, \\
    Y_C \,|\, M &\sim \beta_0 + \beta_2 M^3 + \sigma {\cal N} (0, 1) \,, \nn \\
    Y_T \,|\, M &\sim \beta_0 + \beta_1 + \beta_2 M^3 + \sigma {\cal N} (0, 1) \,. \nn
\end{align}

\begin{figure}[htp!]
\begin{center}
\includegraphics[width=6.5in]{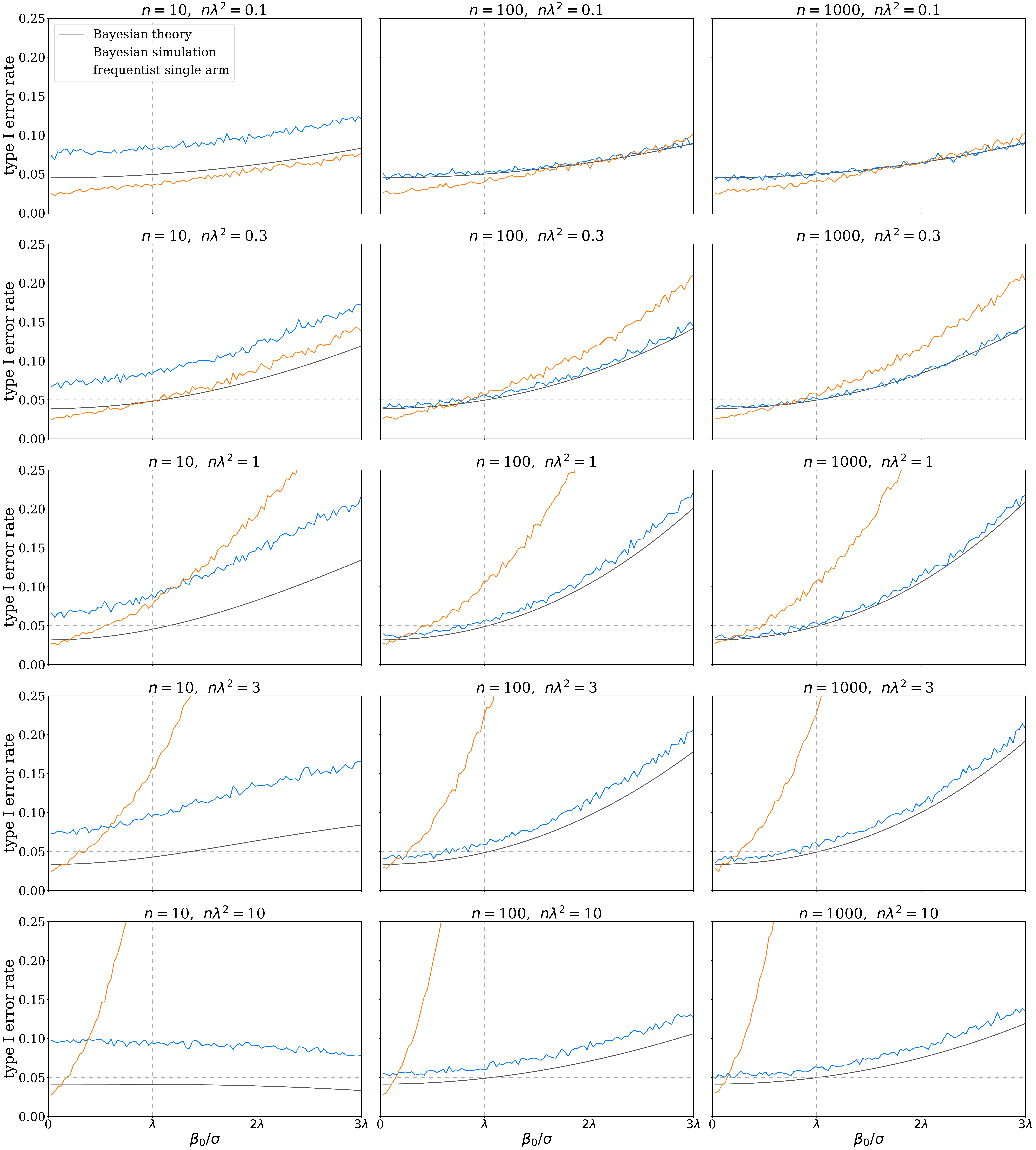}
\end{center}
\caption{{\bf Comparing simulations and theory for type I error when the model assumptions hold.} Type I error rate as a function of $\beta_0 / \sigma$ for various choices of $n$ and $n \lambda^2$.  The type I error rate of the Bayesian estimator defined here is compared between simulations and theory, and against that of the single-arm estimator.  Note that the x-axis is marked in units of $\lambda$, which differs between plots.
\label{fig:gridAdetail}}
\end{figure}

\begin{figure}[htp!]
\begin{center}
\includegraphics[width=6.5in]{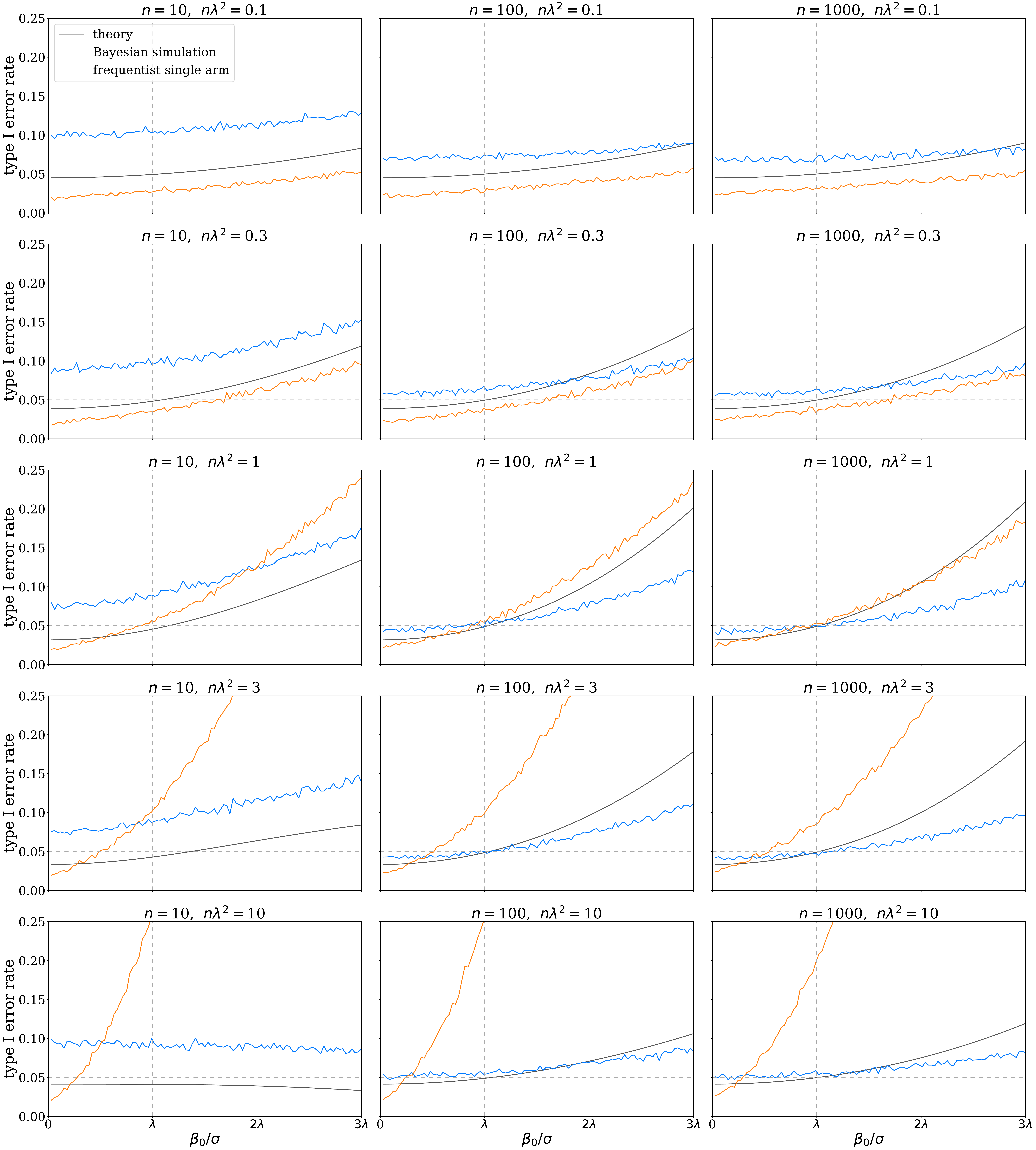}
\end{center}
\caption{{\bf Comparing simulations and theory for type I error when the model assumptions do not hold.} Type I error rate for the same parameters as \Fig{gridAdetail}, except using samples generated in the nonlinear case.  The type I error rate of the Bayesian estimator defined here is compared between simulations and theory, and against that of the single-arm estimator.  Note that the x-axis is marked in units of $\lambda$, which differs between plots.
\label{fig:gridB}}
\end{figure}

In \Fig{gridAdetail}, the type I error rate is shown as a function of $\beta_0 / \sigma$ for various choices of $n$ and $n \lambda^2$ for the linear generative model case.  This is a more general version of \Fig{gridA} in the main text.

When the linearity assumptions of the estimator do not hold, the theory can be inaccurate.  In \Fig{gridB} the type I error rate is shown for the same parameters as \Fig{gridAdetail}, using a nonlinear generative model for prognostic scores and outcomes.  The theoretical predictions for error are large, and as expected the estimator is only appropriate for cases where the linearity assumptions hold.

\subsection{Invariance of Type I error}
\label{sec:type_I_vis}

Previously, we noted how the asymptotic type I error rate of the Bayesian prognostic covariate adjustment estimator does not depend on $\beta_2$ or on $\sigma^2$ except through the ratio $\beta_0/\sigma$. This fact is demonstrated in \Figs{gridC}{gridD}. In each case, the common theory prediction is compared against the dependence of the type I error rate on $\beta_0$ for different choices of the prognostic-outcome correlation (\Fig{gridC}) and $\beta_2$ (\Fig{gridD}). In both cases the theoretical predictions agree well with the simulations, with the largest discrepancy coming at small $n$.

\begin{figure}[htp!]
\begin{center}
\includegraphics[width=5.5in]{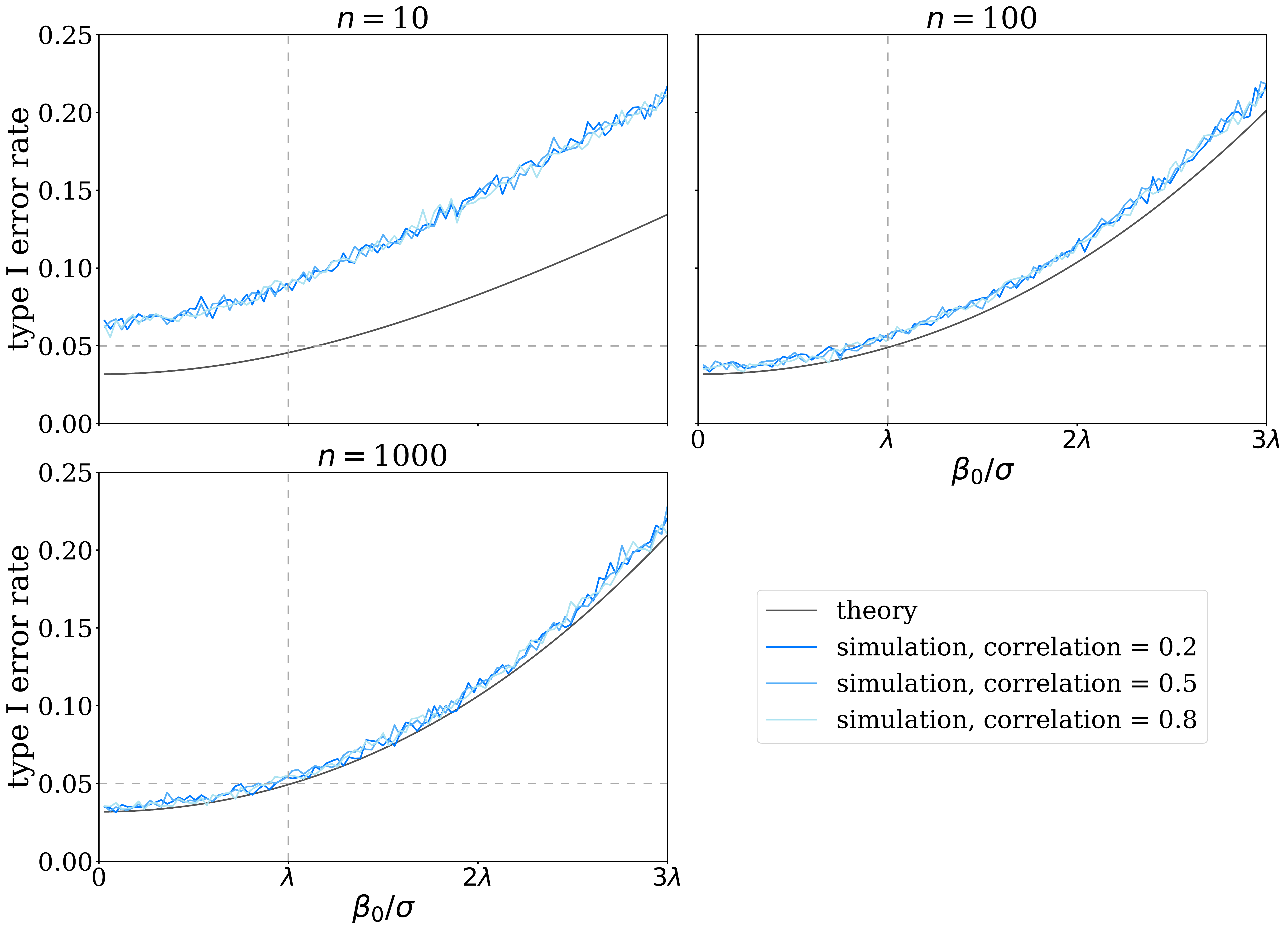}
\end{center}
\caption{{\bf Independence of type I error rate on the prognostic-outcome correlation.} Type I error rate as a function of $\beta_0 / \sigma$ in units of $\lambda$ is shown for several different choices of the number of subjects and the correlation between prognostic scores and outcomes.  Note that the x-axis is marked in units of $\lambda$, which differs between plots as $\lambda$ is held so that $n \lambda^2 = 1$.
\label{fig:gridC}}
\end{figure}

\begin{figure}[tp!]
\begin{center}
\includegraphics[width=5.5in]{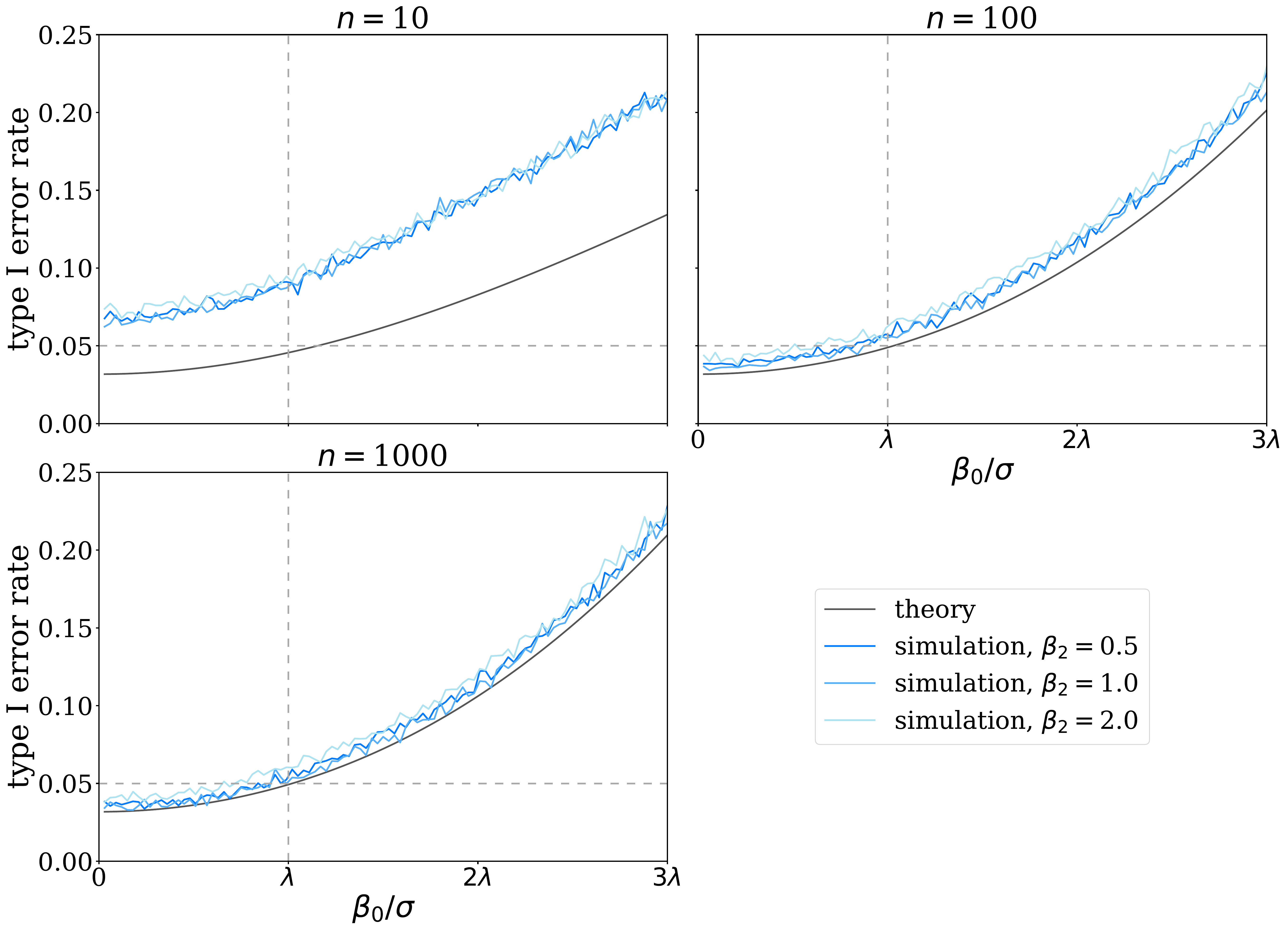}
\end{center}
\caption{{\bf Independence of type I error rate on the value of $\beta_2$.} Type I error rate as a function of $\beta_0 / \sigma$ in units of $\lambda$ is shown for several different choices of the number of subjects and $\beta_2$.  Note that the x-axis is marked in units of $\lambda$, which differs between plots as $\lambda$ is held so that $n \lambda^2 = 1$.
\label{fig:gridD}}
\end{figure}

\subsection{Impact of the prior on power}
\label{sec:power_vis}

\begin{figure}[htp!]
\begin{center}
\includegraphics[width=5.5in]{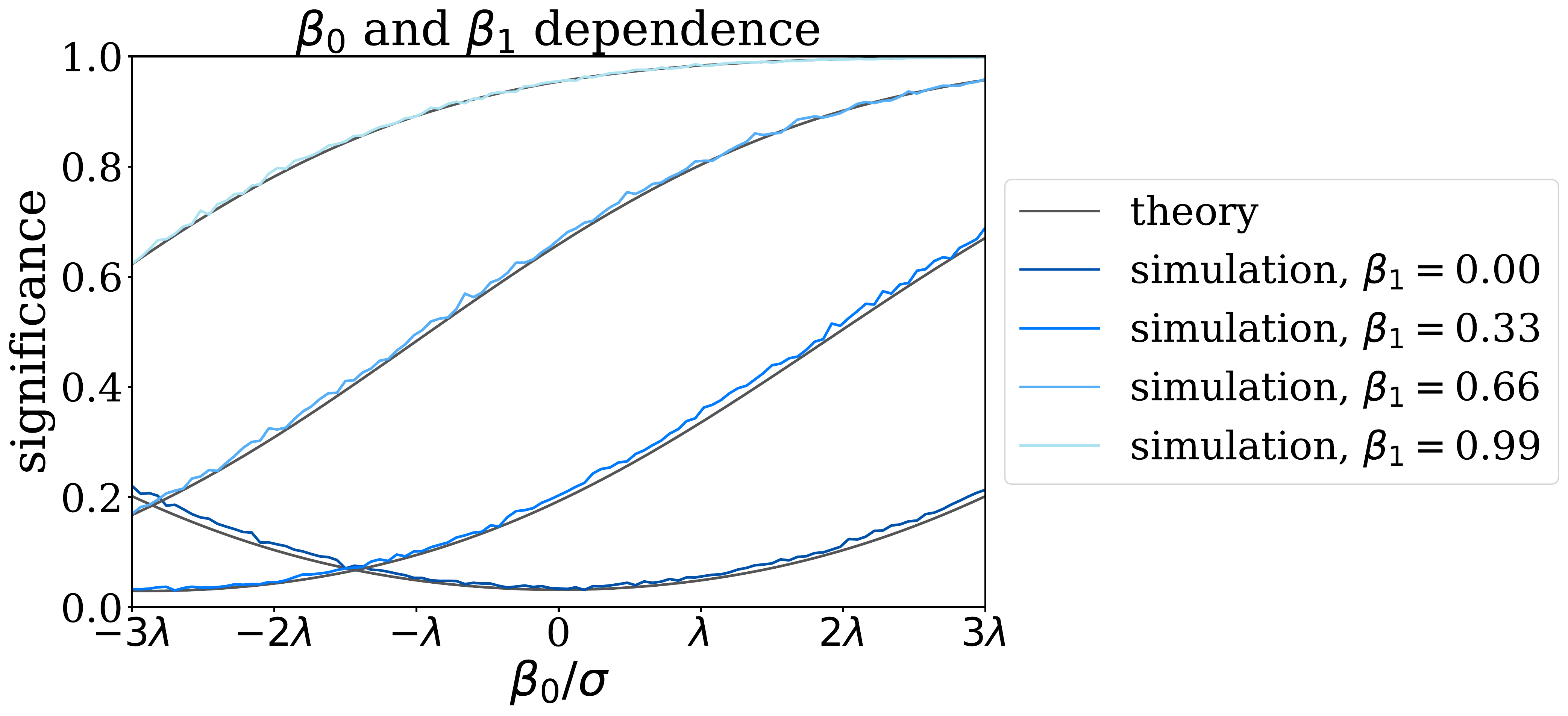}
\end{center}
\caption{{\bf Dependence of significance on $\beta_0$ and $\beta_1$.} Significance is shown as a function of $\beta_0 / \sigma$ in units of $\lambda$ for different values of $\beta_1$. $n = 100$ is used.
\label{fig:gridF}}
\end{figure}

The discussion in the main text highlights how the power depends on the precision of the prior, interpolating between prognostic covariate adjustment when the prior is concentrated tightly around zero and the single-arm analysis when the prior is flatter. In fact the power depends not just on the width of the prior, but also on its accuracy.  In \Fig{gridF} the rejection rate is shown as a function of $\beta_0 / \sigma$ for different choices of $\beta_1$ (including $\beta_1 = 0$, for which this is the type I error rate and is symmetric in $\beta_0$).  The theoretical predictions agree well with the simulations and show that unlike type I error, power is asymmetric in $\beta_0$ depending on whether the prognostic scores are biased in or away from the direction of $\beta_1$.  

\begin{figure}[htp!]
\begin{center}
\includegraphics[width=5.5in]{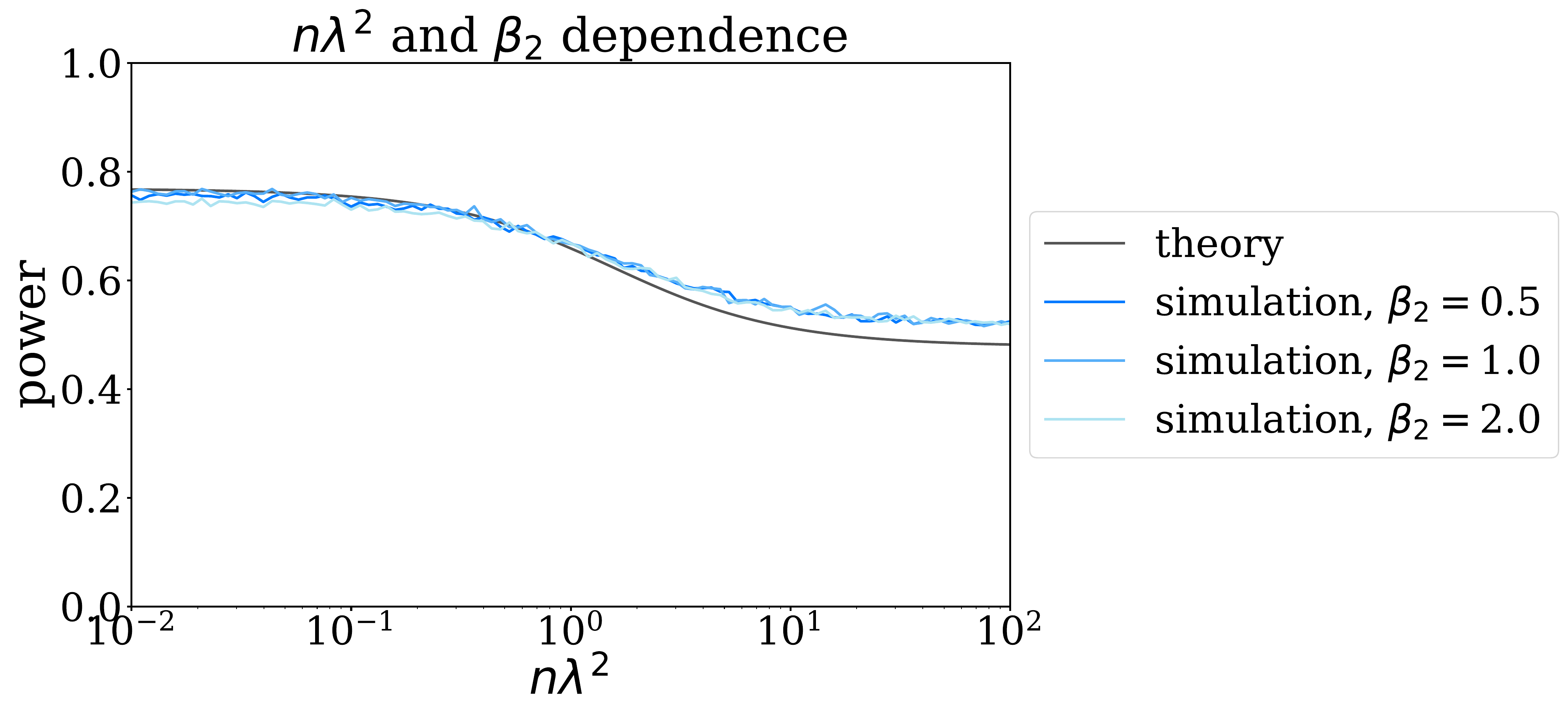}
\end{center}
\caption{{\bf Independence of power on $\beta_2$.} Power as a function of $n \lambda^2$ for different values of $\beta_2$. $n = 100$ is used.  
\label{fig:gridG}}
\end{figure}

Finally, power is predicted to be independent of $\beta_2$, a fact confirmed by simulations. Fig.~\ref{fig:gridG} shows the power as a function of $n \lambda^2$ for different values of $\beta_2$. The theoretical predictions agree well with the simulations except for a slight increase in power at large $n \lambda^2$, and the simulations show that power is independent of $\beta_2$ to a high degree.

\end{document}